\documentclass[12pt, draftclsnofoot, onecolumn]{IEEEtran}
\usepackage{longtable}
\usepackage{color}
\usepackage{soul}
\usepackage{amsmath}
\usepackage{cancel}
\usepackage{amssymb}
\usepackage{lscape}
\usepackage{graphicx}
\usepackage{longtable}
\usepackage{multirow}
\usepackage{accents} 
\usepackage{dblfloatfix}
\usepackage{array}
\usepackage{setspace}
\usepackage{float}
\usepackage{subcaption}
\usepackage{lettrine} 
\usepackage[noadjust]{cite}
\usepackage{balance}
\usepackage[english]{babel}
\usepackage[table,xcdraw]{xcolor}
\usepackage{diagbox}
\usepackage{adjustbox}
\usepackage{caption, multirow, makecell}
\setcellgapes{3pt}
\usepackage{multirow}
\usepackage{algorithm,algpseudocode}
\usepackage{lscape}
\usepackage{suffix}
\usepackage{mathtools}
\usepackage{url} 
\usepackage{cleveref}
\usepackage{soulutf8}
\usepackage{color,soul}

\newtheorem{lemma}{Lemma}
\newtheorem{proposition}{Proposition}

\newtheorem{corollary}{Corollary}
\newenvironment{proof}{{\textit{Proof:}}}{\hfill$\blacksquare$}
\newcommand{\qed}{\hfill $\blacksquare$}

\DeclarePairedDelimiterX\MeijerM[3]{\lparen}{\rparen}%
{#3\delimsize\vert\,\begin{smallmatrix}#1 \\ #2\end{smallmatrix}}

\newcommand\MeijerG[8][]{%
	G^{#2,#3}_{#4,#5}\MeijerM[#1]{#6}{#7}{#8}}

\WithSuffix\newcommand\MeijerG*[7]{%
	G^{#1,#2}_{#3,#4}\MeijerM*{#5}{#6}{#7}}
		
\definecolor{maroon}{RGB}{186,0,0}
\definecolor{purple}{RGB}{96,26,149}
\definecolor{mavi}{RGB}{46,76,255}
\definecolor{haki}{RGB}{38,99,33}

\newcommand{\vect}[1]{\boldsymbol{#1}}

\usepackage{tikz}
\usepackage{mathdots}
\usepackage{cancel}
\usepackage{color}
\usepackage{siunitx}
\usepackage{array}
\usepackage{multirow}
\usepackage{amssymb}
\usepackage{tabularx}
\usepackage{booktabs}
\usetikzlibrary{fadings}
\usetikzlibrary{patterns}
\usetikzlibrary{shadows.blur}
\usetikzlibrary{shapes}
\begin{document} 
\title{\huge UAV-Assisted Cooperative \& Cognitive NOMA: Deployment, Clustering, and Resource Allocation}	
\author{\normalsize
Sultangali~Arzykulov,~\IEEEmembership{\normalsize Member,~IEEE}, 
Abdulkadir~Celik,~\IEEEmembership{\normalsize Senior~Member,~IEEE},
Galymzhan~Nauryzbayev,~\IEEEmembership{\normalsize Member,~IEEE},
and Ahmed~M.~Eltawil,~\IEEEmembership{\normalsize Senior Member,~IEEE}.

\thanks{S. Arzykulov, A. Celik, and A. M. Eltawil are with Computer, Electrical, and Mathematical Sciences \& Engineering (CEMSE) Division at King Abdullah University of Science and Technology (KAUST), Thuwal, KSA 23955-6900 (e-mail: sultangali.arzykulov@gmail.com; \{abdulkadir.celik,ahmed.eltawil\}@kaust.edu.sa).}
\thanks{G. Nauryzbayev is with the Department of Electrical and Computer Engineering School of Engineering and Digital Sciences, Nazarbayev University, Nur-Sultan, Kazakhstan (e-mail: galymzhan.nauryzbayev@nu.edu.kz). }
}
\maketitle
\thispagestyle{empty}
\begin{abstract}
Cooperative and cognitive non-orthogonal multiple access (CCR-NOMA) has been recognized as a promising technique to overcome issues of spectrum scarcity and support massive connectivity envisioned in next-generation wireless networks. In this paper, we investigate the deployment of an unmanned aerial vehicle (UAV) as a relay that fairly serves a large number of secondary users in a hot-spot region. The UAV deployment algorithm must jointly account for user clustering, channel assignment, and resource allocation sub-problems. We propose a solution methodology that obtains user clustering and channel assignment based on the optimal resource allocations for a given UAV location. To this end, we derive closed-form optimal power and time allocations and show it delivers optimal max-min fair throughput by consuming less energy and time than geometric programming. Based on optimal resource allocation, the optimal coverage probability is also provided in closed-form, which takes channel estimation errors, hardware impairments, and primary network interference into account. The optimal coverage probabilities are used by the proposed max-min fair user clustering and channel assignment approaches. The results show that the proposed method achieves 100\% accuracy in more than five orders of magnitude less time than the optimal benchmark.
\end{abstract}
\IEEEpeerreviewmaketitle
\section{Introduction}
\lettrine{T}{he} main requirements of beyond fifth-generation (B5G) wireless networks are typically categorized into three primary service classes \cite{popovski20185g}: enhanced mobile broadband (eMBB) to provide an improved network capacity and peak data rates for high throughput demanding users; massive machine-type communication (mMTC) to support the ever-increasing number of low-power low-cost Internet of things (IoT) devices; and ultra-reliable low-latency (URLLC) communication for mission-critical applications. Optimizing the network resources to achieve these goals jointly is a multi-objective combinatorial problem, which is hard to solve in real-time, even for small-scale networks. The interwoven relations among these goals are coupled by spectral efficiency (SE), which is determined by the interference between users competing for scarce network resources \cite{Celik6G}. 

Legacy orthogonal multiple access (OMA) schemes have struggled to deliver adequate support for eMBB and mMTC service requirements. As a remedy, non-orthogonal multiple access (NOMA) technology has been recognized as a promising technology to improve the SE by simultaneously serving multiple users on the same resource via multiplexing them in either power or code domains \cite{dai2015non}. Coupling NOMA with cooperative communications (CC) and cognitive radios (CR), show significant promise for B5G networks. In CC-NOMA, a relay node with a strong channel decodes the messages intended to weak channel users and fully exploits such prior information to improve the weak user performance \cite{Access_CR_NOMA}. In the CR-NOMA, the unlicensed/secondary users (SUs) operating on the NOMA scheme are permitted to transmit over spectrum bands licensed to primary users (PUs) in an opportunistic and non-intrusive manner \cite{Celik2016Green, Celik2016Multi, Celik2017Hybrid}. Therefore, the conflation of cooperative and cognitive NOMA (CCR-NOMA) concepts has recently attracted much attention to provide dense wireless networks with more significant SE.    

Recently, unmanned aerial vehicles (UAVs) have received attention to serve as an aerial relay to enhance the coverage of geographical regions with high user density and heavy traffic loads, which are also known as (a.k.a.) hot-spots. Since the relay coordinates mainly determine the channel quality and capacity of links to/from the UAV, UAV deployment has a significant impact on the overall network performance. For a given UAV location, CCR-NOMA gain is also affected by the user pairing/clustering strategy, channel allocation due to the varying interference to/from PUs, power control mechanisms, and allocated time portions at each hop. Therefore, this paper investigates the UAV deployment problem in CCR-NOMA networks, where we account for user clustering/pairing, interference to/from a primary network, and resource (i.e., power and time) allocation aspects. 

\subsection{Related Works}
\label{sec:related}
Recent works on CCR-NOMA can be exemplified as follows: In \cite{Li}, the authors derived the closed-form expressions of the outage probability (OP) and ergodic sum-rate (ESR) for full-duplex cooperative NOMA relaying systems with in-phase and quadrature (I/Q) components' imbalance and imperfect successive interference cancellation (SIC). The work \cite{Li2} evaluated the impact of imperfect SIC, non-ideal channel state information, and residual hardware impairments on the OP and ergodic capacity metrics in cooperative NOMA networks over $\alpha-\mu$ fading environment. Due to the fact that cooperative NOMA is susceptible to the inter-user interference (IUI) caused by its operation in the power domain, the spatial modulation was proposed to resolve this issue by avoiding the SIC and IUI terms from the space domain \cite{Chen1, Chen2}. The authors in \cite{tgcn} analyzed the outage performance of underlay CR-NOMA networks with multiple SUs. The throughput maximization problem for the similar system model was solved in \cite{Xu1} by splitting into two subproblems, i.e., NOMA-SU assignment and power allocation, while ensuring the SUs' fairness. The authors in \cite{Xu2} proposed a robust resourc allocation (RA) algorithm to maximize the sum energy efficiency of underlay CR-NOMA networks under channel uncertainties. 

The recent NOMA works on the user clustering and RA can be exemplified as follows: In \cite{Nguyen}, the power allocation, user pairing, and UAV deployment in NOMA networks were jointly studied to maximize the minimum sum-rate per each user pair. The authors in \cite{Celik1} proposed a distributed cluster formation  and RA framework for imperfect NOMA-based interference-limited wireless networks. In addition, in \cite{Celik2}, a similar system model was considered for which the distributed cluster formation and power-bandwidth allocation schemes were proposed. In \cite{Xu}, a joint user clustering and robust beamforming design were proposed to minimize the total transmit power while satisfying the users' quality-of-service (QoS) requirements in downlink NOMA networks with multiple UAV-BSs. A max-min fair UE clustering problem was addressed in \cite{Elkashlan} using three different sub-optimal approaches. The work \cite{Chatzinotas} presented an iterative user clustering method, where each iteration aims at joint optimization of the beamforming and power allocation for given clusters.

The recent works on UAV deployment can be exemplified as follows: in \cite{Mozaffari}, the authors investigated the mobility and optimal deployment of several UAVs to ensure energy-efficient data collection from IoT devices. \cite{Chen} investigated the deployment problem of a single UAV-BS aimed at achieving the maximum reliability performance metrics, such as power losses, bit-error rates, and overall outage. It was demonstrated that the optimal altitude values of both static and mobile UAVs are not identical for different metrics. On the other hand, the authors in \cite{Savkin} optimized the coverage performance of a wireless network with multiple UAV-BSs through minimizing the average distances between the UAV and end-users. The authors in \cite{Alzenad} also studied the wireless network with multiple UAV-BSs and proposed a low-complexity algorithm to solve the deployment problem while maximizing the number of covered end-users with different QoS requirements. In \cite{Mozaffari1}, the optimal three dimensional ($3$D) placement of multiple UAVs deployed with directional antennas was investigated to maximize the total coverage area. The authors in \cite{Kalantari} solved the efficient UAV-BSs' deployment problem for maximizing the coverage performance and defined the minimum number of UAVs to serve all given end-users in a certain area. A stochastic geometry based UAV deployment approaches were developed in \cite{bushnaq2019aeronautical, bushnaq2020optimal}; while the former determined the number and location of the hovering stations to minimize the total time spent for data aggregation from a large scale IoT network, the latter found the optimal locations of tethered and regular UAVs to maximize the coverage of users in a hot-spot region. The authors in \cite{Zeng1,Zeng2,Wu} studied the energy-efficient UAV communication by optimizing the UAV trajectory and transmit power of the UAV while the UAV altitude and bandwidth were optimized by assuming interference-free scenario in \cite{He}.

\subsection{Main Contributions}
\label{sec:contributions}

The main contributions of this paper can be summarized as follows:
\begin{itemize}

\item Considering the effect of system impairments such as channel state information (CSI) and hardware as well as transmission power constraint imposed by the primary network, we analytically derive the end-to-end coverage probability for secondary NOMA users considering the Nakagami-$m$ statistical model.  Unlike free space LoS channel assumption in \cite{Zeng1, Zeng2}, we consider the channel path loss taking into account the effect of both LOS and NLOS components which varies with the UAV location, building density, and height distribution. Finally, Monte Carlo simulations are used to validate the correctness of analytical derivations.

\item For a given set of users (i.e., cluster) on a primary channel, we formulate a fair CCR-NOMA optimization framework for two joint sub-problems: power control and phase-time allocation while taking into consideration the constraints on maximum transmit power of the secondary base station and UAV as well as an interference constraint imposed by the primary network. We provide closed-form max-min power and phase-time allocations, which are validated with numerical results obtained by geometric programming. The obtained results show that closed-form approaches provide more energy-efficient solutions much faster than geometric programming. The optimal end-to-end coverage probability is then obtained by using the optimal RA values.  

\item By using optimal end-to-end probabilities, we propose a fast yet highly accurate user clustering and channel allocation approach, which maximizes the minimum coverage probability of the secondary network. Unlike the common min-sum assignment (a.k.a. Hungarian Algorithm) approach used for the maximum sum-rate objective, we exploit the linear bottleneck assignment (LBA) approach. The numerical results show that the LBA achieves $100$\% accuracy at five orders of magnitude faster time duration than an optimal integer linear programming benchmark. The proposed clustering and channel allocation approach are also suitable for real-time applications as it provides solutions for $200$ users and channels in around $64$ milliseconds.

\item Based on the above steps, we lastly investigate an optimal $3$D deployment of the UAV to maximize the minimum throughput and coverage performance of the CCR-NOMA network. That is, the proposed deployment takes all user clustering, channel assignment, and resource allocation subproblems into account. In this regard, our contributions are distinct from others that examine the 2D positioning of the UAV without altitude optimization \cite{Wu, Zeng1} or do not consider inter-user/inter-network interference \cite{Alzenad, He}.

\end{itemize}

\subsection{Notations and Paper Organization}
\label{sec:notations}
Throughout the paper, sets and their cardinality are denoted with calligraphic and regular uppercase letters (e.g., $| \mathcal{A}|=A$), respectively. Vectors and matrices are represented in lowercase and uppercase boldfaces (e.g., $\vect{a}$ and $\vect{A}$), respectively. The i$^{th}$ row vector of $\vect{A}$ is denoted by $\vect{A}_i$. Subscripts $p$ and $s$ refers to the primary and secondary base stations, respectively. Likewise, subscripts $k$ and $n$ is used for indexing primary users/channels and secondary users, respectively. The subscript $r$ represents the relaying UAV. The notation $a_b^c$ denotes the parameter/variable $a$ from $b$ to $c$, $(a,b) \in \{ p, s, k, n, r,\}$. For example, $d_s^r$/$h_s^r$ denotes the distance/channel from the secondary base station to the relaying UAV. This notation is also extended to $a_{b,c}^{d,e}$ to describe parameters/variables related to different transmitters, receiver, channels, and destination. 

The remainder of the paper is organized as follows. Section \ref{sec:sys} presents the considered network and channel models for the UAV-assisted CCR-NOMA network. Section \ref{sec:problem} discusses the problem formulation and introduces the proposed solution methodology. Furthermore, Section \ref{sec:perform_analysis} derives new analytical expressions for the coverage probability over Nakagami-$m$ fading channels while Section \ref{sec:SP1} provides the closed-form max-min fair resource allocation derivations. Then, Section \ref{sec:SP23} proposes the LBA-based user clustering, channel assignment, and UAV deployment approach. Lastly, Section \ref{sec:res} presents numerical and simulation results and Section \ref{sec:conc} concludes the paper by remarking the key findings.

\section{System Model}
\label{sec:sys}
\subsection{Network Model}
\label{sec:network}
We consider a downlink CCR-NOMA network that consists of a primary network (PN) and a secondary network (SN), as illustrated in Fig. \ref{fig:UAV_sys_model}. The PN comprises of a single primary BS (PBS) that serves $K$ primary users (PUs) over $K$ primary channels (PCs) in an orthogonal and time-slotted fashion. At each time-slot duration of $T$,  the PBS transmits on PC$_k$ with power $P_p^k$ such that overall power consumption cannot exceed the total transmit power $P_p$, i.e., $\sum_{k=1}^K P_p^k \leq P_p$. The set of SUs allocated to the same PC is referred to as a cluster and denoted by $\mathcal{C}_k=\{n \: \vert  \:  \chi_k^n=1 \}$, where $\chi_k^n \in \{0,1\}$ is the binary channel allocation indicator variable. The cluster size is represented by $C_k \triangleq \vert \mathcal{C}_k \vert=\sum_n \chi_k^n, \: \forall k$. Each cluster/PC/PU has a bandwidth of $W$ Hz\footnote{There is a one-to-one correspondence between $\mathcal{C}_k$ and PC$_k$. Thus, they are interchangeably used throughout the paper.}.  The PN operator allows the SN to use PC$_k$ in a cognitive underlay manner such that the SN ensures that its total interference to PU$_k$ on PC$_k$ cannot exceed a predetermined interference temperature constraint (ITC) threshold \cite{tvt}, $\text{ITC}_k$. 
\begin{figure}[t]
	\centering
	\includegraphics[width=0.7\textwidth]{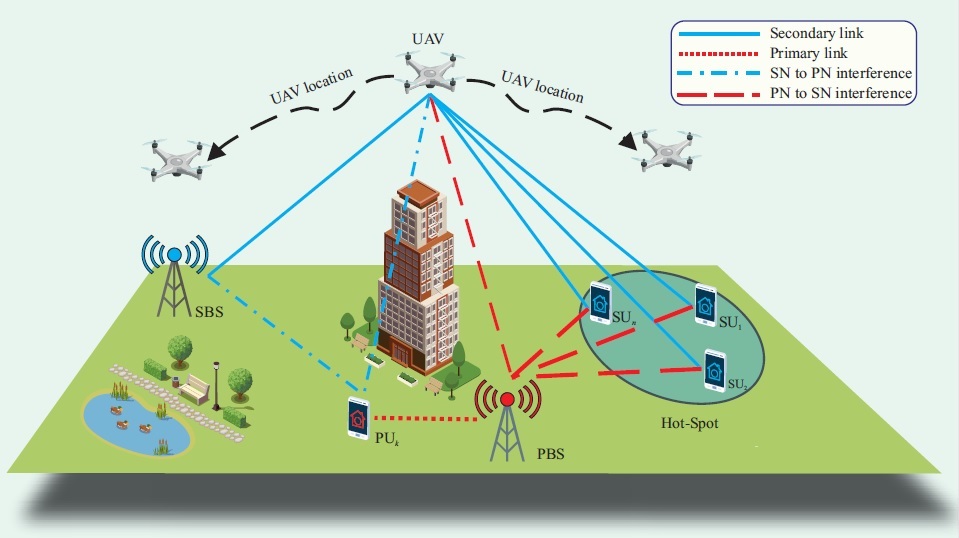}
	\caption{The illustration of the considered underlay CCR-NOMA network.} 
	\label{fig:UAV_sys_model}
\end{figure}

On the other hand, the SN consists of a secondary base station (SBS), a relaying UAV, and SUs. The SBS is overloaded with geographical regions that experience high-user density and heavy traffic conditions, which are also referred to as \textit{hot-spots}. In order to keep up with the growing QoS demands of $N$ SUs located within the hot-spot, the SBS is required to improve the hot-spot coverage by offloading some of its traffic onto PCs while ensuring that ITCs are not violated. With the aim of enhancing the traffic offloading performance, the network operator deploys a UAV that hovers around the hot-spot region and relays the hot-spot traffic from the SBS. Hence, each time-slot is divided into broadcasting and relaying phases with durations of $\lambda_k T$ and $(1-\lambda_k) T$, respectively. In the first phase, the SBS transmits on PC$_k$ with power $P_s^k$ such that overall power consumption cannot exceed the total transmit power $P_s$, i.e., $\sum_{k=1}^K P_s^k \leq P_s$. Likewise, in the second phase, the UAV transmits on PC$_k$ with power $P_r^k$ such that overall power consumption cannot exceed the total transmit power $P_r$, i.e., $\sum_{k=1}^K P_r^k \leq P_r$\footnote{In the remainder of the paper, we assume that the total downlink transmission power of the PBS/SBS/UAV are evenly distributed among channels/clusters. Thus, we have $P_p^k=P_p/K$, $P_s^k=P_s/K$, and $P_r^k=P_r/K$ for the PBS, SBS, and UAV, respectively.}.

\subsection{Channel Model}
\label{sec:channel}
For a generic transmitter node $i$ and receiver node $j$, the composite channel gain is given by 
\begin{equation}
\label{eq:comp_ch}
g_i^j= \sqrt{\ell_i^j} h_i^j, \: i \in \{ p, s, r\} , \: j \in \{k, n ,r\}, \: i\neq j,
\end{equation}
where $\ell_i^j$ is the spatial expectation of attenuation over the probabilities of having line-of-sight (LoS) and non-line-of-sight (NLoS) links, and $h_i^j$ is the channel gain that represents small scale fading. Denoting $\upsilon_i^j(\rm{LoS})\triangleq 1-\upsilon_i^j(\rm{NLoS})$ as the probability of having a LoS link between $i$ and $j$, the spatial expectation of attenuation factor is given by 
\begin{equation}
\label{eq:spatial_ch}
\ell_i^j=\prod_{l\in\{\rm{LoS}, \rm{NLoS}\}}\left[ \eta_i^l \rm{FSPL} \right]^{-\upsilon_i^j(l)}
\end{equation}
where $\eta_i^l$ is the attenuation coefficient related to LoS/NLoS links, and \rm{FSPL} is the free-space path loss. Similarly, the spatial expectation of the path loss (measured in dB) is given by \cite{Hourani}
\begin{equation}
\label{eq:spatial_PL}
\text{PL}_i^j=\sum_{l\in\{\rm{LoS}, \rm{NLoS}\}} \upsilon_i^j(l) \left( \text{FSPL} +  \eta_i^l \right) ,
\end{equation}
and \rm{FSPL} is calculated as  $\text{FSPL} =20 \log_{10}(d_i^j) + 20  \log_{10}\left(f_c\right) + 20 \log_{10}\left(\frac{4\pi}{c}\right)$,
where $f_c$ is carrier frequency and $c$ is the speed of light. 
For the fixed PBS/SBS heights $H_p$/$H_s$, the probability of having an LoS transmission on the SBS-UAV backhaul link and PBS-UAV interference link is given by
\begin{align}
\label{LoS_i_r}
	\upsilon^{r}_{i} (\rm{LoS}) &= \frac{1}{ 1+a_i \exp\left[-b_i \left(\arctan\left(\frac{H_r - H_i}{d^r_i} \right) - a_i  \right)  \right] },
\end{align}
where $\forall i \in \{p,s\}$, $a_i$ and $b_i$ are approximation parameters depending on $H_i$, building heights distribution, the ratio of land area covered by buildings to total land area, and the mean number of buildings per km$^2$ \cite{ITU}. Similarly, assuming a zero height for SU$_n$, the LoS probability for UAV-SU$_n$ access link and PBS-SU$_n$ interference link can be derived as
\begin{align}
\label{LoS_i_N}
\upsilon^n_j (\rm{LoS}) &= \frac{1}{ 1+a_n \exp\left[-b_n \left(\arctan\left(\frac{H_j}{d^n_j} \right) - a_n  \right)  \right]},
\end{align}
where $\forall j \in \{p,r\}$. Assuming Nakagami-$m$ small scale fading, $\vert h_i^j \vert^2$ follows the Gamma distribution with the following probability density function (PDF) \cite{simon_alouini_2005}  
\begin{align}
f_{\vert h^i_j\vert^2} (x) = \frac{m^m x^{m-1} \exp\left[-m x\right]}{\Gamma (m)}, 
\end{align}
where $\Gamma(\cdot)$ is the gamma function and $m$ is the fading parameter, from which Rayleigh fading channel can be modeled by setting $m = 1$ and the Rician channel is approximated by setting $m>1$. To capture CSI imperfections, we model channel coefficients using the minimum mean square error channel estimator as $h_i^j = \tilde{h}_i^j + e_i^j$,  where $\tilde{h}_i^j$ and $e_i^j\sim  \mathcal{CN}(0,\zeta_i^j)$ are the estimated channel coefficient and channel estimation error, respectively. The error variance is modeled as $\zeta_i^j \triangleq \theta \rho^{- \mu}$, where $\rho = \frac{P}{\sigma^2}$ is the transmitted signal-to-noise ratio (SNR), and $\mu \geqslant 0 $, $\theta > 0$  \cite{impCSI}. This model allows us to describe the SNR-dependent and independent imperfect CSI scenarios by setting $\mu \neq 0$ and $\mu = 0$, respectively. 
\section{Problem Statement and Proposed Solution Methodology}
\label{sec:problem}
In this section, we first present a formal problem statement and then discuss an outline of the proposed solution methodology. 
\subsection{Problem Statement}
Our objective is to constitute a max-min fair coverage enhancement for SUs residing within the hot-spot region. Accordingly, the objective of interest can be expressed as 
\begin{equation}
\label{eq:problem_statement}
\underset{\vect{c}, \vect{\chi}, \vect{\lambda}, \vect{\alpha}, \vect{\beta}}{\max} \left( \underset{\forall k, \forall n}{\min} \left[ \wp_k^n (\vect{\chi}_k) \chi_k^n \right] \right), 
\end{equation}
where $\vect{c} \in \mathbb{R}^3$ is the UAV location; $\vect{\chi} \in \{0,1\}^{K \times N}$ is the assignment matrix; $\vect{\lambda} \in \mathbb{R}^{K}$ is the vector of phase duration portions; $\vect{\alpha} \in \mathbb{R}^{K \times N}$ / $\vect{\beta} \in \mathbb{R}^{K \times N}$ is the matrix of power allocation factors in the broadcasting/relaying phases; and $\wp_k^n (\vect{\chi}_k)$ is the coverage probability of SU$_n \in \mathcal{C}_k$. The standard formulation of the max-min coverage problem can be given as follows
\begin{equation}
\hspace*{0pt}
 \begin{aligned}
 & \hspace*{0pt} \vect{\mathrm{P}_\mathrm{o}}: \underset{\vect{c}, \vect{\chi}, \vect{\alpha}, \vect{\beta}, \vect{\lambda}, \psi}{\max}
& & \hspace*{3 pt} \psi \\
& \hspace*{0pt} \mbox{$\mathrm{C_o^1(k,n)}$: }\hspace*{20pt} \text{s.t.}
 &&   \chi_k^n \wp_k^n (\vect{\chi}_k)+(1-\chi_k^n)
 \geq \psi , \textbf{ } \forall k,n\\ 
 &
 \hspace*{0 pt}\mbox{$\mathrm{C_o^2(n)}$: } & & \sum_{k} \chi_k^n = 1, \textbf{ }  \forall n \\
    &
     \hspace*{0 pt}\mbox{$\mathrm{C_o^3(k)}$: } & & \sum_{n} \chi_k^n \leq \lceil N/K \rceil, \textbf{ }  \forall k \\
    &
\hspace*{0 pt}\mbox{$\mathrm{C_o^4(k)}$: } & &  \sum_{n}  \alpha_k^n \chi_k^n \leq \min \left\{ 1,\frac{\text{ITC}_k}{P_s^k g_s^k} \right\}, \: \forall k\\
    &
    \hspace*{0 pt}\mbox{$\mathrm{C_o^5(k)}$: } & &  \sum_{n} \beta_k^n \chi_k^n \leq \min \left\{ 1,\frac{\text{ITC}_k}{P_r^k g_r^k} \right\} ,\: \forall k\\
  &
\hspace*{0 pt}\mbox{$\mathrm{C_o^6(k,n)}$: } & & \vect{c} \in \mathbb{R}^3, \chi_k^n \in \{0,1\}, \\
  &
 & &  \psi, \lambda_k, \beta_k^n, \alpha_k^n  \in [0,1], \forall k, \forall n
\end{aligned}
\end{equation}
where the max-min objective is handled by setting the objective to an auxiliary variable $\psi$ and enforcing all coverage probabilities to be greater than or equal to $\psi$. Thus, the constraint in $\mathrm{C}_0^1$ is a standardized formulation of the following constraint 
\begin{equation}
\mathrm{C}_o^1=
    \begin{cases}
    \wp_k^n (\vect{\chi}_k) \geq \psi &, \text{ if } \chi_k^n=1, \\
    1-\wp_k^n (\vect{\chi}_k) \geq \psi &, \text{ if } \chi_k^n=0,
    \end{cases}
\end{equation}
where $\wp_k^n (\vect{\chi}_k)$ is zero if $\chi_k^n=0$. The constraint $\mathrm{C}_o^2$ allows an SU to join a single cluster at a time while the constraint $\mathrm{C}_o^3$ limits the cluster size to $\lceil N/K \rceil$. The constraint $\mathrm{C}_o^4$ dictates two fundamental constraints on the total SBS power allocation on PC$_k$: 1) The total power allocation cannot exceed the maximum permissible power on PC$_k$ $\left(\sum_{n}  \alpha_k^n \chi_k^n \leq 1 \right)$; and 2) The interference received by PU$_k$ cannot exceed the ITC threshold $\text{ITC}_k$ $\left( \sum_{n}  \alpha_k^n \chi_k^n \leq \frac{\text{ITC}_k}{P_s^k g_s^k} \right)$. Both power constraints also apply to the UAV as shown in $\mathrm{C}_o^5$. By setting $\psi$ to a constant value $\overline{\psi}$, $\vect{\mathrm{P}_\mathrm{o}}$ can be reduced to a feasibility problem, where a coverage probability of no less than $\overline{\psi}$ is guarantied for all SUs. Since $\vect{\mathrm{P}_\mathrm{o}}$ is a mixed integer non-linear programming problem (MINLP), it falls within the class of $\rm{NP}$-Hard problems. Since reaching the optimal solution takes impractically long times even for moderate size of networks, it is necessary to develop a fast yet efficient a sub-optimal approach, which is explained in the sequel. 

\begin{figure}[t]
	\centering
	\includegraphics[width=0.5\textwidth]{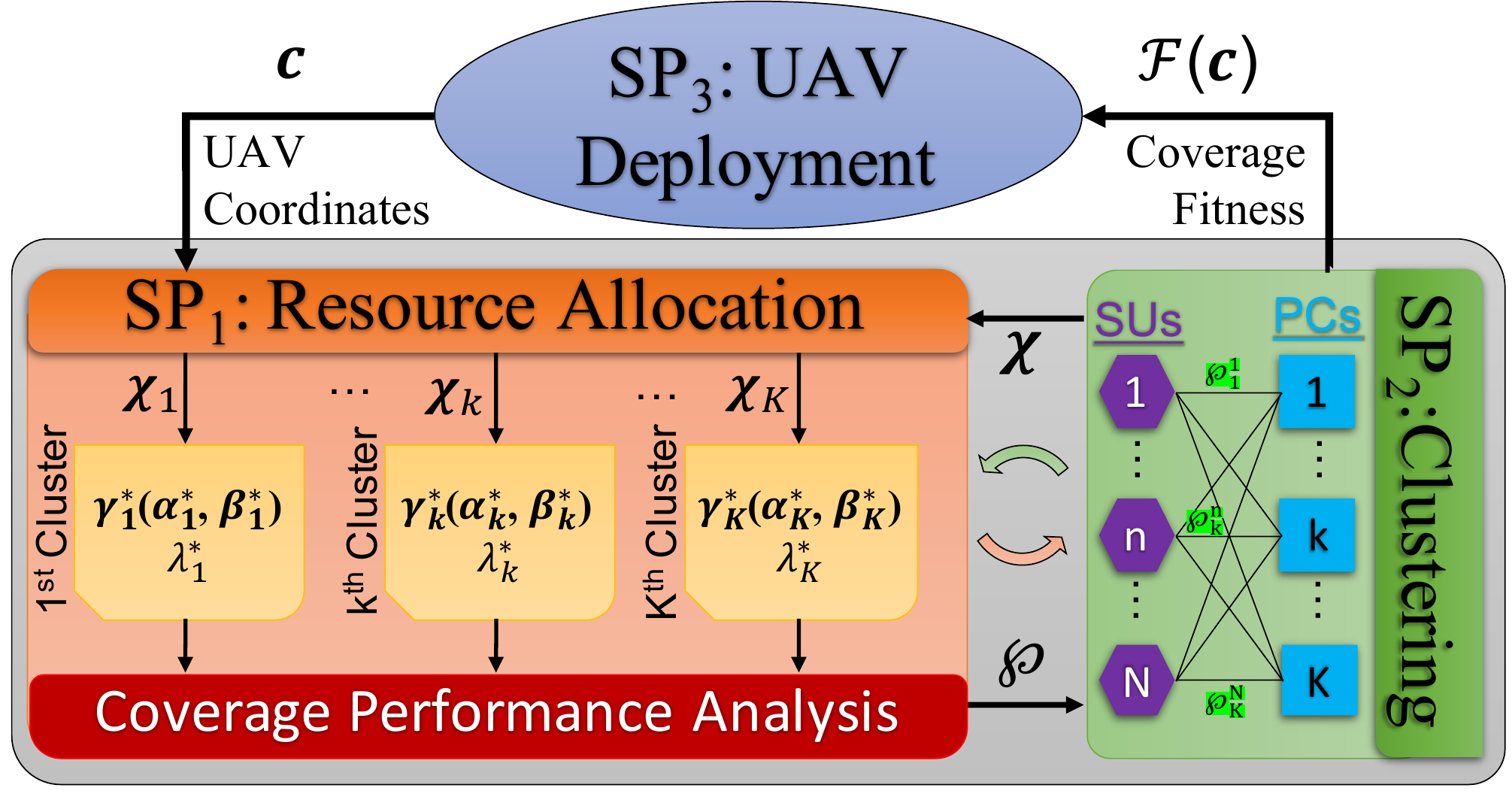}
	\caption{Schematic illustration of the solution methodology.} 
	\label{fig:sol_meth}
\end{figure}

\subsection{Solution Methodology}
As illustrated in Fig. \ref{fig:sol_meth}, $\vect{\rm{P_o}}$ can be decomposed into the following three sub-problems: 
\begin{enumerate}
    \item[$\vect{\mathrm{SP_1}}$] 
    Resource allocation is a joint power control and phase-time allocation sub-problem that takes the UAV location $\vect{c}$ and the assignment matrix $\vect{\chi}$ for granted. The $\vect{\mathrm{SP_1}}$ optimizes the power allocation factors of the broadcasting and relaying phases (i.e., $\accentset{\star}{\vect{\alpha}}_k$ and $\accentset{\star}{\vect{\beta}}_k$) to obtain the optimal SIDNRs of $\mathcal{C}_k$, $\accentset{\star}{\vect{\gamma}}_k\left(\accentset{\star}{\vect{\alpha}}_k,\accentset{\star}{\vect{\beta}}_k\right)$. Based on $\accentset{\star}{\vect{\gamma}}_k$, the phase-time allocation calculates the optimal  time allocations, $\accentset{\star}{\lambda}_k$, which yields the max-min fair data rate for $\mathcal{C}_k$, $\accentset{\star}{R}_k$. In Section \ref{sec:SP1}, we provide closed-form solutions for $\accentset{\star}{\vect{\gamma}}_k$, $\accentset{\star}{\lambda}_k$, $\accentset{\star}{\vect{\alpha}}_k$, and $\accentset{\star}{\vect{\beta}}_k$. The optimal max-min fair coverage probability of $\mathcal{C}_k$ is obtained by substituting $\accentset{\star}{\vect{\gamma}}_k$ and $\accentset{\star}{\lambda}_k$ into the closed-form cumulative distribution function (CDF) derived in Section \ref{sec:perform_analysis}. Resulting coverage probability matrix $\vect{\wp} \in \mathbb{R}^{K \times N}$ is then fed into the clustering sub-problem. 
    \item [$\vect{\mathrm{SP_2}}$]  
    For a given UAV location $\vect{c}$ and coverage probability matrix $\vect{\wp}$, $\vect{\mathrm{SP_2}}$ assigns SUs to clusters and clusters to PCs. Then, it returns the coverage fitness of the UAV location $\mathcal{F}(\vect{c})$ (i.e., the max-min SNR of the entire secondary network) to the deployment sub-problem. The proposed clustering algorithm is developed in Section \ref{sec:SP23}. 
    
    \item [$\vect{\mathrm{SP_3}}$] 
    Lastly, the UAV deployment problem leverages $\vect{\rm{SP_1}}$ and $\vect{\rm{SP_2}}$ to evaluate the coverage fitness of UAV locations. To do so, $\vect{\mathrm{SP_3}}$ feeds a UAV location $\vect{c}$ and get a coverage fitness feedback $\mathcal{F}(\vect{c})$ in return. In this way, the UAV deployment problem can alter the coordinates to find the best location that provides the max-min fair coverage to the entire hot-spot region. In Section \ref{sec:SP23}, Algorithm \ref{alg:SP123} presents the algorithmic implementation of overall solution methodology.
\end{enumerate}

\section{Coverage Performance Analysis of UAV-Assisted CCR-NOMA}
\label{sec:perform_analysis}
Without loss of generality, let us first focus our attention on $\mathcal{C}_k$/PC$_k$ to explain the relaying mechanisms and then analyze the coverage probability for each cluster member.   

\subsection{UAV-Assisted CCR-NOMA}
\subsubsection{The Broadcasting Phase}
In this phase, the SBS superposes the messages intended to $SU_n \in \mathcal{C}_k, \forall n$ and broadcasts the resultant signal to the UAV. The signal broadcast to the UAV is given by $x_k= \sum_{n \in \mathcal{C}_k} \sqrt{\alpha_k^n} s_n$, where $\alpha_k^n$ and $s_n$ is the power allocation factor of the first phase and the corresponding message intended for SU$_n$. By accounting for both CSI imperfections and hardware impairments, the UAV receives the broadcast signal as follows
\begin{align}
\label{signal at R}
y_{s,r}^{k} =& h_s^r\sqrt{\varrho_s^r} \left( x_k +  \phi^r_s \right)+ h^r_p \sqrt{\varrho_p^r} \left(s^{}_p + \phi^r_p\right) + n^{}_r,  
\end{align}
where $\varrho_i^r = P_i \ell_i^r / K$, $i \in \{s,p\}$, $\eta^{}_{(\cdot)} \sim \mathcal{CN}\left(0, \phi^{2}_{(\cdot)} \right)$ denotes the aggregate distortion noise from transceiver; $\phi^{}_{(\cdot)}$ = $\sqrt{\phi^{2}_{t}+\phi^{2}_{r}}$ is the aggregate hardware impairment (HI) level from the transmitter and receiver \cite{Stefania}, and $n^{}_{(\cdot)} \sim \mathcal{CN}\left(0, \sigma^{2}_{(\cdot)} \right)$ is the additive white Gaussian noise (AWGN) at each receive node. From  \eqref{signal at R}, the instantaneous signal-to-interference-distortion-noise-ratio (SIDNR) to decode the message $s_n$ at the UAV can be expressed by
\begin{align}
\label{eq:gamma_knsr}
\gamma_{k,n}^{s,r} =\frac{\vert \tilde{h}^r_s \vert^2 \alpha^n_k}{\vert \tilde{h}^r_s \vert^2 I_{k,n}^{s} + \vert\tilde{h}^r_s \vert^2 \sigma^2_{\phi^r_s}  +E^r_s  + \vert h^r_p \vert^2 I^r_p +   \bar{\sigma}^2_r },
\end{align} 
where the SIC interference is given by $I_{k,n}^{s}=  \sum_{j=n+1}^{C_k} \alpha^j_k$ and $ \sigma^2_{\phi^r_s}$ stands for the hardware distortion noise power whereas $E^r_s =\sigma^2_{e^r_s} \left(1+\sigma^2_{\phi^r_s} \right)$ is the power of the channel error and the interference received from the PBS is given by $I^r_p = \frac{\varrho_p^r}{\varrho_s^r}  \left(1+\sigma^2_{\phi^r_p}\right)$. Lastly, $\bar{\sigma}^2_r= \sigma^2_r/\varrho_s^r$ is the normalized thermal noise at the UAV receiver. 

\subsubsection{The Relaying Phase}
In the phase, the UAV relays the decoded information by broadcasting $\tilde{x}_k= \sum_{n \in \mathcal{C}_k} \sqrt{\beta_k^n} \tilde{s}_n$, where $\beta_k^n$ and $\tilde{s}_n$ are the power allocation factor of the second phase and message dedicated for SU$_n$, respectively. Hence, during the relaying phase, the user $n$ receives the following signal
 \begin{align}
 \label{signal at n}
 y^{k}_{r,n} =& h^n_r \sqrt{\varrho_r^n} \left( \tilde{x}^{}_k +  \phi^n_r \right) + h^n_p \sqrt{\varrho_p^n} \left(s^{}_p + \phi^r_p\right) + n^{}_n,  
 \end{align}
where $\varrho_i^n = P_i \ell_i^n$ / K, $i \in \{r,p\}$. From  \eqref{signal at n}, the instantaneous SIDNR to decode the message $s_n$ at SU$_n$ can be expressed by  
\begin{align}
\label{eq:gamma_knr}
\gamma_{k,n}^{r,n} =\frac{\vert \tilde{h}^n_r \vert^2 \beta_k^n}{\vert \tilde{h}^n_r \vert^2 I_{k,n}^{r} + \vert \tilde{h}^n_r \vert^2 \sigma^2_{\phi^n_r}  + E^n_r + \vert h^n_p \vert^2 I^n_p  + \bar{\sigma}^2_n },
\end{align}
where  
the SIC interference is denoted by $I_{k,n}^{r}= \sum_{i=b+1}^{n} \beta^i_k$ while $  \sigma^2_{\phi^n_r}$ stands for the power of the hardware distortion noise; $E^n_r =  \sigma^2_{e^n_r} \left(1+\sigma^2_{\phi^n_r} \right) $ is the power of the channel error; $I^n_p = \frac{\varrho_p^n}{\varrho_r^n}  \left(1+\sigma^2_{\phi^n_p} \right)  $ is the interference received from the PBS; and $\bar{\sigma}^2_n= \sigma^2_n/\varrho_r^n$ is the normalized thermal noise at the user $n$. 

\subsubsection{End-to-End SIDNRs and Data Rates}
Following from \eqref{signal at R} and \eqref{signal at n}, the end-to-end SIDNR of SU$_n$ within $\mathcal{C}_k$ is given by
\begin{equation}
\label{eq:e2e_snr}
\gamma_k^n= \min \left\{ \gamma_{k,n}^{s,r}, \gamma_{k,n}^{r,n} \right\}, \forall n \in \mathcal{C}_k,
\end{equation}
based on which the max-min SIDNR of $\mathcal{C}_k$ is given by $\gamma_k= \underset{n \in \mathcal{C}_k }{\min} \left(\gamma_k^n\right), \forall k$. Similarly, the end-to-end data rate of SU$_n$ within $\mathcal{C}_k$ is given by
\begin{align}
\label{eq:e2e_rate}
R_k^n = \min &\left\{ \lambda_k W  \log_2\left(1+\gamma_{k,n}^{s,r}\right),  (1-\lambda_k) W \log_2 \left(1+ \gamma_{k,n}^{r,n} \right) \right\},
\end{align}
based on which the max-min rate of $\mathcal{C}_k$ is given by $R_k= \underset{n \in \mathcal{C}_k }{\min} \left(R_k^n\right), \forall k$.

\subsection{Coverage Probability Analysis}
\label{sec:Coverage Analysis}
The coverage probability is defined as the likelihood of having an achievable rate no less than a predefined threshold $\bar{R}$. Thus, the coverage probability of SU$_n \in \mathcal{C}_k$ can be expressed by
\begin{align}
\label{Poutn}
 {\rm{Pr}} \left[ R_k^n \geq \bar{R} \right]
 &\stackrel{(a)}{=} {\rm{Pr}} \left[ \min \left\{ \lambda_k W  \log_2\left(1+\gamma_{k,n}^{s,r}\right),  (1-\lambda_k) W \log_2 \left(1+ \gamma_{k,n}^{r,n} \right) \right\} \geq \bar{R} \right] \nonumber \\
& \stackrel{(b)}{=} {\rm{Pr}} \left[\gamma_{k,n}^{s,r}  \geq  \bar{\gamma}_k^1 \right] \text{Pr}\left[\gamma_{k,n}^{r,n}  \geq \bar{\gamma}_k^2\right],
\end{align} 
where $(a)$ follows from \eqref{eq:e2e_rate}, $(b)$ follows from the assumption of independent broadcasting and relaying channels, $\bar{\gamma}_k^s \triangleq 2^{\frac{\bar{R}}{\lambda_k W}}-1$, and $\bar{\gamma}_k^r \triangleq 2^{\frac{\bar{R}}{(1-\lambda_k) W}}-1$. Next, we derive the closed-form coverage probabilities in the first and second phases, i.e.,  $\wp_{k,n}^{s,r}$ and  $\wp_{k,n}^{r,n}$, respectively.
\begin{lemma}
\label{lem:CDF}
    In the first and second phases, the coverage probability of SU$_n \in \mathcal{C}_k$ over LoS and/or NLoS Nakagami-$m$ fading channels are respectively given by 
	\begin{align}
	\label{coverage_prob_r}
{\rm{Pr}}\left[ \gamma_{k,n}^{s,r} \geq \bar{\gamma}_k^1 \right] &= 1-{\rm{Pr}}\left[ \gamma_{k,n}^{s,r} < \bar{\gamma}_k^1 \right]=  1 - F_{\gamma_{k,n}^{s,r}}(\bar{\gamma}_k^s),\\
{\rm{Pr}}\left[ \gamma_{k,n}^{r,n} \geq \bar{\gamma}_k^2 \right] &= 1-{\rm{Pr}}\left[ \gamma_{k,n}^{r,n} < \bar{\gamma}_k^2 \right]=  1 - F_{\gamma_{k,n}^{r,n}}(\bar{\gamma}_k^r),
	\end{align} 
where the CDF for both phases, $F_{\gamma_{k,n}^{i,j}}(\bar{\gamma}_k^i)$, with $i \in \{s,r\}$, $j \in \{r,n\}$, $i\neq j$,  is derived as follows
	\begin{align}
	\label{CDFsr_fin}
	F_{\gamma_{k,n}^{i,j}} (\bar{\gamma}_k^i) =&  \frac{\gamma_{inc}\left( {y_i^k}, {y_i^k} \Lambda_i \right)}{\Gamma({y_i^k})} - \frac{\gamma_{inc}\left( {y_i^k}, {y_i^k} \Lambda_i \right)}{\Gamma({y_i^k})} \frac{ \left({z_p^j} \right)^{{z_p^j}} \exp\left[-{x_i^j} \left(\mathcal{E}_i  + \mathcal{S}_i \right) \right] }{\Gamma({z_p^j})} 
	\sum_{q=0}^{{x_i^j}-1}
	\frac{\left({x_i^j} \right)^q}{q!} \sum_{l=0}^{q} \binom{q}{l} \nonumber\\
	& \times  \frac{ \left( \mathcal{E}_i  + \mathcal{S}_i \right)^q (\mathcal{I}_i)^{l} \Gamma({z_p^j} + l)}{\left({z_p^j} + {x_i^j} \mathcal{I}_i  \right)^{{z_p^j}+l} } + \frac{\Gamma\left( {y_i^k}, {y_i^k} \Lambda_i \right) }{\Gamma({y_i^k})} - \frac{ \left({z_p^j} \right)^{{z_p^j}}
 \exp\left[- \Lambda_i \left({y_i^k} + {x_i^j} \mathcal{V}_i  \right) \right]}{\Gamma({z_p^j}) }   \nonumber\\
	&\times \frac{\left({y_i^k} \right)^{{y_i^k}}
 \exp\left[-{x_i^j} \mathcal{E}_i\right] }{(\mathcal{U}_i)^{z_p^j} \, \Gamma({y_i^k})} \sum_{u=0}^{\jmath} \left( \mathcal{V}_i\right)^{\jmath-u}  \sum_{\Bbbk=0}^{{x_i^j}-1} \frac{\left
 ({x_i^j} \right)^{\Bbbk-{z_p^j} - u } }{\Bbbk !} \sum_{\jmath=0}^{\Bbbk} \binom{\Bbbk}{\jmath} \left(\mathcal{E}_i\right)^{\Bbbk-\jmath} \sum_{p=0}^{{y_i^k} + \jmath-1}   \frac{\left(\Lambda_i \right)^p }{p !}  \nonumber \\
	&\times \sum_{t=0}^{p} \binom{p}{t}  \left({y_i^k} + {x_i^j} \mathcal{V}_i  \right)^{p-u-\jmath - y_i^k - z_p^j}       \MeijerG*{2}{1}{1}{2}{1-({z_p^j} + u+t)}{0,~-{z_p^j} + {y_i^k}-t}{\frac{ {z_p^j} + {x_i^j} \Lambda_i  \mathcal{U}_i   }{{x_i^j}  \mathcal{U}_i \left({y_i^k} + {x_i^j}\mathcal{V}_i \right)^{-1} }}
	\end{align}
which reads the terms from Table \ref{tab:CDF_terms}.	
\end{lemma}
\begin{proof}
	 See Appendix \ref{Appendix 1}.
\end{proof}
\begin{table}[t]
\centering
\caption{List of Symbols and Notations for Eqs. \eqref{CDFsr_fin} and \eqref{CDFsr}. }
\label{tab:CDF_terms}
\resizebox{1\textwidth}{!}{
	\begin{tabular}{|l|l|l|l|l|l|}
		\hline
		\multicolumn{3}{|l|}{$X_i^j = \vert \tilde{h}^j_i \vert^2$, $i \in \{s,r\}$, $j \in \{r,n\}$, $i\neq j$}                                                                                                       & \multicolumn{3}{l|}{Absolute power of estimated channel between nodes $i$ and $j$}                                                                                                                                           \\ \hline
		\multicolumn{3}{|l|}{$Y_i^k = \vert g^i_k \vert^2$, $i \in \{s,r\}$}                                                                                                                                           & \multicolumn{3}{l|}{Absolute power of interference channel from $i$ to $k$}                                                                                                                                                  \\ \hline
		\multicolumn{3}{|l|}{$Z_p^j = \vert h_p^j \vert^2$, $j \in \{r,n\}$}                                                                                                                                           & \multicolumn{3}{l|}{Absolute power of interference channel from PBS to $j$}                                                                                                                                                  \\ \hline
		\multicolumn{3}{|l|}{$\mathcal{A}_i = v^n_k - I_{k,n}^{i} \bar{\gamma}_k^i - \sigma^2_{\phi^j_i} \bar{\gamma}_k^i$}                                                                                            & \multicolumn{3}{l|}{where $v = \alpha$ for $i =s$ and $v = \beta$ for $i=r$}                                                                                                                                                 \\ \hline
		\multicolumn{3}{|l|}{$I^j_{\bar{p}} = \frac{\varrho_p^j}{\text{ITC}_k \ell_i^j} \left(1+\sigma^2_{\phi^j_p}\right)$}                                                                                  & \multicolumn{3}{l|}{Interference received from the PBS considering the ITC}                                                                                                                                                  \\ \hline
		\multicolumn{3}{|l|}{$\bar{\sigma}^2_{\bar{j}}= \frac{\bar{\sigma}^2_j}{\text{ITC}_k \ell_i^j}$}                                                                                                            & \multicolumn{3}{l|}{Normalized thermal noise at the UAV receiver considering the ITC}                                                                                                                                        \\ \hline
		\multicolumn{3}{|l|}{${x_i^j}$, ${y_i^k}$ and ${z_p^j}$}                                                                                                                                                 & \multicolumn{3}{l|}{Fading parameters of channels $X_i^j$, $Y_i^k$ and $Z_p^j$, respectively}                                                                                                                                \\ \hline
		\multicolumn{3}{|l|}{$\gamma_{inc}\left( \cdot \right)$}                                                                                                                                                       & \multicolumn{3}{l|}{Lower incomplete Gamma function}                                                  \\ \hline \hline
		~~$\mathcal{I}_i = \frac{I^j_p \bar{\gamma}_k^i}{\mathcal{A}_i}$~ & ~~$\mathcal{E}_i = \frac{E^j_i \bar{\gamma}_k^i}{\mathcal{A}_i}$~ & ~~$\mathcal{S}_i = \frac{\bar{\sigma}^2_j \bar{\gamma}_k^i}{\mathcal{A}_i}$~ & ~~~~~$\mathcal{V}_i = \frac{ \bar{\gamma}_k^i \bar{\sigma}^2_{\bar{j}} }{\mathcal{A}_i }$~~~~~ &~~~~~ $\mathcal{U}_i = \frac{ \bar{\gamma}_k^i I^j_{\bar{p}} }{\mathcal{A}_i }$~~~~~ & ~~~~~~$\Lambda_i = \frac{\text{ITC}_k}{\bar{P}_i^k}$~~~~~~ \\ \hline
	\end{tabular}
}
\end{table}

\section{Max-Min Fair Resource Allocation}
\label{sec:SP1}
The resource allocation problem consists of two joint sub-problems: power control and phase-time allocation. Since clusters have dedicated and independent power and time resources, they do not share any conflicting variables. This paves the way for the further decomposition of $\vect{\mathrm{SP_1}}$ into individual cluster resource allocation problems [c.f. Fig. \ref{fig:sol_meth}]. Each cluster can be further decoupled into power control and phase-time allocation sub-problems as they involve two distinct and independent resources. Therefore, in what follows, we focus our attention on a generic cluster $\mathcal{C}_k$ without loss of generality. 

\subsection{Power Control}
\label{sec:PA}
The coverage probability and data rate maximizations are two equivalent problems because the coverage probability is a monotonically increasing function of the data rate. Considering the complexity of the CDF function derived in Lemma \ref{lem:CDF}, alternating to the max-min rate problem is preferable for the sake of tractability. Noting that $\gamma_{k,n}^{s,r}$ and $\gamma_{k,n}^{r,n}$ are not functions of $\lambda_k$ [c.f. \eqref{eq:gamma_knsr} and \eqref{eq:gamma_knr}], the optimal power allocation that gives the maximum $\gamma_k^n$ (thus $R_k^n$) is independent of $\vect{\lambda_k}$, which justifies the motivation behind decoupling explained above. For a given UAV location and cluster set, the equivalent problem can be formulated as follows   \textbf{}
\begin{equation}
\label{CVX}
\hspace*{0pt}
 \begin{aligned}
 & \hspace*{0pt} \vect{\mathrm{SP}}_1^k (\mathcal{C}_k, \vect{c}): \underset{\vect{\alpha_k}, \vect{\beta_k}, \gamma_k}{\max}
& & \hspace*{3 pt} \gamma_k \\
& \hspace*{43pt} \mbox{$\mathrm{C_1^1}$: }\hspace*{15pt} \text{s.t.}
&&   \gamma_k^n \geq \gamma_k , \textbf{ } \forall  n \in  \mathcal{C}_k\\ 
 &
\hspace*{43pt}\mbox{$\mathrm{C_1^2}$: } & &  \sum_{n \in \mathcal{C}_k}  \alpha_k^n \leq \min \left\{ 1,\frac{\text{ITC}_k}{P_s^k g_s^k} \right\}\\
    &
    \hspace*{43pt}\mbox{$\mathrm{C_1^3}$: } & &  \sum_{n \in  \mathcal{C}_k} \beta_k^n  \leq \min \left\{ 1,\frac{\text{ITC}_k}{P_r^k g_r^k} \right\}\\
  &
\hspace*{43pt}\mbox{$\mathrm{C_1^4}$: } & & \beta_k^n, \alpha_k^n, \gamma_k  \in [0,1],
\end{aligned}
\end{equation} 
where $\gamma_k$ is an auxiliary variable similar to $\psi$ in $\vect{\mathrm{P_o}}$. $\vect{\mathrm{SP}}_1^k$ can be numerically solved by geometric programming \cite{boyd2007tutorial} via altering $\mathrm{C}_1^1$ into $1/\gamma_k^n \leq 1/\gamma_k$ to put inequality constraints in the form of posynomials, as in $\mathrm{C}_1^2$ and $\mathrm{C}_1^3$. Fortunately, $\vect{\mathrm{SP}}_1^k$ can also be solved analytically based on two key propositions: 

\begin{proposition}
\label{prop1}
Following from \eqref{eq:e2e_snr}, the constraint $\mathrm{C}_1^1$ can be expanded into two set of constraints: $\gamma_{k,n}^{s,r} \geq \gamma_k$ and $\gamma_{k,n}^{r,n} \geq \gamma_k, \: \forall k$. This intuitively dictates at the optimal point that the SIDNRs of both phases must be no less than the optimal SIDNR, $\accentset{\star}{\gamma}_k, \: \forall k$.
\end{proposition}
\begin{proposition}
\label{prop2}
At the optimal point ($\accentset{\star}{\vect{\alpha}}_k, \accentset{\star}{\vect{\beta}}_k, \accentset{\star}{\gamma}_k$), at least one of the SIDNR constraints must be active, i.e., $\gamma_{k,n}^{s,r} = \gamma_k \vee \gamma_{k,n}^{r,n} = \accentset{\star}{\gamma}_k, \: \exists  k$. That is, there is no unique solution and thus enforcing all SIDNRs to be equal to $\accentset{\star}{\gamma}_k$ is still optimal. Indeed, this equalization is especially preferable since it is also optimal in terms of the total energy consumption (i.e., $E_k= \sum_{n \in \mathcal{C}_k}  \alpha_k^n + \sum_{n \in \mathcal{C}_k}  \beta_k^n$). This is true due to the fact that feasible SIDNRs greater than $\accentset{\star}{\gamma}_k$ require higher power-allocation factors, and thus higher $E_k$. 
\end{proposition}
Based on Propositions \ref{prop1} and \ref{prop2}, we have $2 C_k$ equalities from $\mathrm{C}_1^1$ and 2 inequalities from $\mathrm{C}_1^2-\mathrm{C}_1^3$. These are sufficient to find a closed-form primal solutions for  $2 C_k + 1$ variables in the sequel, which are given in the sequel.
\begin{lemma}
\label{lem:SP1}
The optimal first-phase power allocation factors for SU$_n \in \mathcal{C}_k$ is given by 
\begin{equation}
    \label{eq:alpha_star}
    \accentset{\star}{\alpha}_k^n=  I_r \accentset{\star}{\vect{\gamma}} \left(1+\accentset{\star}{\vect{\gamma}}    \right)^{N-n},
\end{equation}
where $I_r =   \sigma^2_{\phi^r_s}  + \frac{E^r_s  + \vert \tilde{h}^r_p \vert^2 I^r_p +\bar{\sigma}^2_r}{\vert\tilde{h}^r_s \vert^2} $ is the interference plus noise term of the SIDNR in the broadcasting phase. Further, by substituting $\accentset{\star}{\alpha}_k^n$ into \eqref{eq:gamma_knsr}, the optimal first-phase SIDNRs for SU$_n \in \mathcal{C}_k$ is given by 
\begin{equation}
    \label{eq:gamma_star_1}
    \accentset{\star}{\gamma}_{k,n}^{s,r}= \left( \frac{ \Phi_n}{I_r} +1\right)^{\frac{1}{N}} -1.
\end{equation}
The optimal second phase power allocation factors for SU$_n \in \mathcal{C}_k$ is given by 
\begin{equation}
    \label{eq:beta_star}
    \accentset{\star}{\beta}_k^n= I_n \accentset{\star}{\vect{\gamma}} + \accentset{\star}{\vect{\gamma}}^2  \sum_{j=n+1}^{N} I_j (1+\accentset{\star}{\vect{\gamma}})^{j-n-1},
\end{equation}
where $ I_n =  \sigma^2_{\phi^r_n}  + \frac{E^r_n  + \vert \tilde{h}^n_p \vert^2 I^n_p +   \bar{\sigma}^2_n}{\vert\tilde{h}^r_n \vert^2}$ is the interference plus noise term of the SIDNR in the relaying phase. Then, substituting $ \accentset{\star}{\beta}_k^n$ into \eqref{eq:gamma_knr}, the optimal second-phase SIDNRs for SU$_n \in \mathcal{C}_k$ is given by\footnote{Here, we present the derivation of optimal SIDNR for $N=2$ as the general optimal SIDNR equation becomes indefinable for $N>2$ due to the lack of a recognizable pattern. } 
\begin{equation}
    \label{eq:gamma_star_2}
    \accentset{\star}{\gamma}_{k,n}^{r,n}= 
    \accentset{\star}{\gamma}_2= \frac{ \sqrt{4 I_2 \Phi_2  + \left(I_1 + I_2 \right)^2 } - I_1 -I_2 }{2 I_2}.  
\end{equation}
\end{lemma}
\begin{proof}
Please see Appendix \ref{app:SP1}. 
\end{proof}

\begin{corollary}
\label{cor:gamma_star}
The end-to-end max-min SIDNR of SU$_n \in \mathcal{C}_k$, $\accentset{\star}{\gamma}_k^n$, can be obtained by substituting $\accentset{\star}{\gamma}_{k,n}^{s,r}$ and $\accentset{\star}{\gamma}_{k,n}^{r,n}$ into \eqref{eq:e2e_snr}, which yields $\accentset{\star}{\gamma}_k^n=\accentset{\star}{\gamma}_{k,n}^{s,r}=\accentset{\star}{\gamma}_{k,n}^{r,n}, \: \forall n \in \mathcal{C}_k$. Accordingly, the max-min SIDNR for $\mathcal{C}_k$ is given by $\accentset{\star}{\gamma}_k= \underset{n \in \mathcal{C}_k }{\min} \left(\accentset{\star}{\gamma}_k^n \right), \forall k$. 
\end{corollary}
\begin{proof}
This corollary directly follows from Proposition \ref{prop2} and Lemma \ref{lem:SP1}. 
\end{proof}

\subsection{Phase-Time Allocation}
\label{sec:TA}
For given SIDNRs, the phase-time allocation sub-problem $\underset{\lambda_k}{\max} (\wp_k^n)$ can be equivalently expressed as $\underset{\lambda_k}{\max}  (R_k^n)$, whose closed-form solution is provided as follows:
\begin{lemma}
\label{lem:SP1_lambda}
For given SIDNRs, the max-min fair phase-time allocation factor for $\mathcal{C}_k$ is given by 
\begin{equation} 
   \accentset{\star}{\lambda}_k\overset{(a)}{=}\frac{\log_2\left(1+\accentset{\star}{\gamma}_{k,n}^{r,n}\right)}{\log_2\left(1+\accentset{\star}{\gamma}_{k,n}^{s,n}\right)+\log_2\left(1+\accentset{\star}{\gamma}_{k,n}^{r,n}\right)}\overset{(b)}{=}\frac{1}{2}.
\end{equation}
\end{lemma}
\begin{proof}
For given SIDNR values, data rates of the first (second) phase increases (decreases) with increasing $\lambda_k$. Because of this inverse relation, the max-min coverage probability is achieved when both phases deliver the same data rate, which yields equality (a). Equality (b) follows from Corollary \ref{cor:gamma_star} which states $\accentset{\star}{\gamma}_k^n=\accentset{\star}{\gamma}_{k,n}^{s,r}=\accentset{\star}{\gamma}_{k,n}^{r,n}, \: \forall n \in \mathcal{C}_k$.  
\end{proof}

Accordingly, the max-min fair data rates and coverage probabilities can be obtained as described in the following corollaries.
\begin{corollary}
\label{cor:R_star}
The end-to-end max-min data rate of SU$_n \in \mathcal{C}_k$, $\accentset{\star}{R}_k^n$, can be obtained by substituting $\accentset{\star}{\lambda}_k$, $\accentset{\star}{\gamma}_{k,n}^{s,r}$, and $\accentset{\star}{\gamma}_{k,n}^{r,n}$ into \eqref{eq:e2e_rate}. Therefore, the max-min rate for $\mathcal{C}_k$ is given by $\accentset{\star}{R}_k= \underset{n \in \mathcal{C}_k }{\min} \left(\accentset{\star}{R}_k^n \right), \forall k$.
\end{corollary}
\begin{proof}
This corollary directly follows from Lemma \ref{lem:SP1} and Lemma \ref{lem:SP1_lambda}.
\end{proof}
\begin{corollary}
\label{cor:P_star}
The max-min coverage probability of SU$_n \in \mathcal{C}_k$, $\accentset{\star}{\wp}_k^n$, can be obtained by substituting $\accentset{\star}{R}_k^n$ into \eqref{CDFsr_fin}. Therefore, the max-min coverage probability for $\mathcal{C}_k$ is given by $\accentset{\star}{\wp}_k= \underset{n \in \mathcal{C}_k }{\min} \left(\accentset{\star}{\wp}_k^n \right), \forall k$.
\end{corollary}
\begin{proof}
This corollary directly follows from Lemma \ref{lem:CDF} and Corollary \ref{cor:R_star}.
\end{proof}

\section{User Clustering, Channel Assignment, and UAV Deployment}
\label{sec:SP23}
In this section, we first present the proposed user clustering and channel asssignment approach, then provide the algorithmic implementation deployment and overall solution methodology. 

\subsection{Clustering: An Iterative Bottleneck Assignment Approach}
The user clustering and channel assignment sub-problem can be formulated as follows
\begin{equation}
\hspace*{0pt}
 \hspace*{0pt} \vect{\mathrm{SP}}_2 \left(\vect{c}\right): \underset{\vect{\chi}}{\max}
 \hspace*{3 pt} \underset{\forall k, \forall n}{\min} \left[\wp_k^n (\vect{\chi}_k) \chi_k^n \right] \text{ s.t. } \mathrm{C_o^{2}},\: \mathrm{C_o^{3}},\: \mathrm{C_o^{6}},
\end{equation}
which is an MINLP problem. Notice that the mixed nature is caused by the term $\wp_k^n (\vect{\chi}_k)$ that varies with the cluster size and cluster member identities. The special case of $\lceil N/K\rceil=1$ reduces the objective to $\underset{\vect{\chi}}{\max} \left(\underset{\forall k, \forall n}{\min} \left\{P_k^n \chi_k^n\right\}\right)$, where $P_k^n$ is the max-min coverage probability of $\mathcal{C}_k$ if SU$_n$ is admitted to $\mathcal{C}_k$. Unlike the cost metric $\wp_k^n (\vect{\chi}_k)$, $P_k^n$ is just a constant rather than being a function of $\vect{\chi}$ as clusters can admit only one SU at a time. Thus, setting $\lceil N/K\rceil=1$ converts clustering problem into an integer linear programming (ILP) SU-PC assignment problem. An alternative solution approach to $\vect{\mathrm{SP}}_2$ is iteratively running $\lceil N/K\rceil$ ILP problems such that each iteration adds an extra member to clusters until all SUs are assigned to a PC.
\begin{proposition}
The coverage performance of a cluster is determined by its members, not by the order of member admissions to the cluster set, i.e., $\min\{x,y\}=\min\{y,x\}$. Since the current ILP iteration optimally admits new cluster members based on the cluster sets formed optimally in the previous iterations, ILP iterations are expected to yield an optimal MINLP solution at the very end.   
\end{proposition}
The mathematical representation of this ILP SU-PC assignment problem is also known as LBA problem \cite{BurkardAssignProb} and can be expressed as in line \ref{line:GBAP_formula} of Algorithm \ref{alg:SP123}. Indeed, LBA is the max-min version of the well-known min-sum (i.e., $\sum_{k,n}P_k^n \chi_k^n$) generalized assignment problem (GAP). In light of the above discussions, we present the proposed clustering solution between line \ref{line:SP2_begin} and line \ref{line:SP2_end} of Algorithm \ref{alg:SP123}. Lines \ref{line:init_p} \& \ref{line:init_I} initialize cluster coverage probabilities, cluster sets, assignment variables, and number of iterations, respectively. At each and every iteration of the most outer loop between lines \ref{line:loop_I_begin} \& \ref{line:loop_I_end}, the outer loop between lines \ref{line:loop_K_begin} \& \ref{line:loop_K_end} and inner loop between lines \ref{line:loop_N_begin} \& \ref{line:loop_N_end} generate the cost matrix $\vect{P}(i)$ as follows: If SU$_n$ is already a member of any cluster, line \ref{line:Pkn1_set} forbids its admission to $\mathcal{C}_k$ by setting $P_k^n$ to zero. Else if SU$_n$ is not a member of any cluster, line \ref{line:Pkn2_set} hypothetically admits SU$_n$ into $\mathcal{C}_k$ and evaluate the new cluster probability as explained in Section \ref{sec:SP1}, which is then set to $P_k^n$. Line \ref{line:call_GBAP} executes SU-PC assignment by calling the LBA between lines \ref{line:GBAP_begin} \& \ref{line:GBAP_end}. Then, clusters and their coverage probabilities are updated in line \ref{line:update_C} and line \ref{line:update_p}, respectively. Once the most outer loop is terminated, overall max-min fair coverage probability and assignment matrix are set in line \ref{line:set_p} and line \ref{line:set_chi}, respectively. 

 \begin{algorithm}[t!]
\footnotesize
 \caption{\textbf{: Deployment-Clustering-Resource Allocation}}
  \label{alg:SP123}
\begin{algorithmic}[1]
 \renewcommand{\algorithmicrequire}{\textbf{Input:}}
 \renewcommand{\algorithmicensure}{\textbf{Output:}}
\State \hspace{-15pt} \textbf{Input:} Environmental parameters
\State $\vect{c} \gets $ Initialize the UAV location \label{line:init_c}
\State $\mathcal{F} \gets \rm{SP_2}(\vect{c}(0))$ // Evaluate the initial location
\State $\accentset{\star}{\vect{c}} \gets \vect{c}$ // Set the best location
\State $\accentset{\star}{\mathcal{F}} \gets  \mathcal{F}$ // Set the best coverage fitness
\For{t=1:T}
\State $\vect{c}(t) \gets$ Randomly pick a neighbor location
\State $\mathcal{F}(t) \gets \rm{SP_2}(\vect{c}(t))$  // Evaluate the neighbor location
\State $\rm{Temp}\gets \textsc{CalculateTemperature}(t,T)$ 
\If {$\mathcal{F}(t) \geq \mathcal{F}(t-1)$} 
\State $\mathcal{F} \gets \mathcal{F}(t)$
\If{$\mathcal{F}(t) \geq \accentset{\star}{\mathcal{F}}$}
\State $\accentset{\star}{\mathcal{F}} \gets \mathcal{F}(t)$
\EndIf
\ElsIf{$\exp\left\{\frac{\mathcal{F}(t)-\mathcal{F}}{\rm{Temp}}\right\}>\rm{rand}$}
\State $\mathcal{F} \gets \mathcal{F}(t)$
\EndIf
\EndFor
\\
\Return $\accentset{\star}{\vect{c}}$, $\accentset{\star}{\mathcal{F}}$
\vspace{2pt}
\hrule 
\vspace{2pt}
\Procedure{$\rm{SP_2}$}{$\vect{c}$} \label{line:SP2_begin}
\State \hspace{-15pt} \textbf{Input:} $\vect{c}$
\State $\mathfrak{p}_k(0) \gets 1$ // Initialize the cluster coverage probabilities, $\forall k$. \label{line:init_p}
\State $\mathcal{C}_k(0) \gets \emptyset$ // Initialize clusters sets, $\forall k$. \label{line:init_C}
\State $\chi_k^n(0) \gets 0$ // Initialize SU-Cluster assignments, $\forall k, \forall n$. \label{line:init_chi}
\State $I \gets \lceil N/K \rceil$ \label{line:init_I}
\For{i=1:I} \label{line:loop_I_begin}
\For{k=1:K} // Generate the $K \times N$ cost matrix, $\vect{P}$. \label{line:loop_K_begin}
\For{n=1:N} \label{line:loop_N_begin}
\If{$\chi_k^n(i-1)=1, \: \exists k$}
\State $P_k^n(i) \gets 0$ \label{line:Pkn1_set}
\ElsIf{$\chi_k^n(i-1)=0, \: \forall k$} \label{line:Pkn2}
\State $P_k^n(i) \gets \rm{SP_1} \left(\mathcal{C}_k(i-1) \cup n , \vect{c} \right)$ \label{line:Pkn2_set}
\EndIf
\EndFor \label{line:loop_K_end}
\EndFor \label{line:loop_N_end}
\State $\vect{\chi}(i) \gets $\textsc{Linear Bottleneck Assignment}($\vect{P}(i)$) \label{line:call_GBAP}
\State $\mathcal{C}_k(i) \gets \{n \vert \chi_k^n(i)=1, \forall n\}$ Update clusters, $\forall k$. \label{line:update_C}
\State $\mathfrak{p}_k(i) \gets \min \{P_k^n(i)\chi_k^n(i)\vert \chi_k^n(i)=1, \forall n\}$ Update cov. prob. \label{line:update_p}
\EndFor \label{line:loop_I_end}
\State $\accentset{\star}{\mathfrak{p}} \gets \underset{\forall k}{\min} \left\{\mathfrak{p}_k\left(I \right)\right\}$ \label{line:set_p}
\State $\accentset{\star}{\vect{\chi}} \gets \vect{\chi}(I)$ \label{line:set_chi} \\
\Return $\accentset{\star}{\mathfrak{p}}$, $\accentset{\star}{\vect{\chi}}$ \label{line:return_p_chi}
\EndProcedure \label{line:SP2_end}
\Procedure{Linear Bottleneck Assignment}{$\vect{P}$} \label{line:GBAP_begin}
\State \hspace{-15pt} \textbf{Input:} $\vect{P}$
\State $\accentset{\star}{\vect{\chi}} \gets \underset{\vect{\chi}}{\max} \left(\underset{\forall k, \forall n}{\min} \left\{P_k^n \chi_k^n\right\}\right) \text{s.t.} \sum_n \chi_k^n = 1, \sum_k \chi_k^n = 1$ \label{line:GBAP_formula}
\Return $\accentset{\star}{\vect{\chi}}$ 
\EndProcedure \label{line:GBAP_end}
 \end{algorithmic}
 \end{algorithm}

One way of solving the LBA is using the threshold approach\footnote{Another way of solving LBA is using augmenting path method that mimics renown Hungarian algorithm, whose complexity is given by $\mathcal{O}\left(K N \sqrt{M \log M}\right)$ \cite{gabow-tarjan}.},  which has a complexity of $\mathcal{O}\left(M^{2.5}/\sqrt{\log_2 M}\right)$ for an $M \times M$ cost matrix \cite[Theorem 6.4]{BurkardAssignProb}. Denoting $I=\lceil N/K \rceil$ and $M=\max\{K,N\}$, the overall time complexity of the proposed clustering approach is given by 
\begin{equation}
\mathcal{O}\left(I \left[\frac{K \times N \times M^{2.5}}{\sqrt{\log_2 M}}\right]\right) \overset{(K=N)}{\approx} \mathcal{O}\left(\frac{N^{4.5}}{\sqrt{\log_2 N}}\right).
\end{equation}
where $K \times N$ is the complexity of generating the cost matrix. In Section \ref{sec:res}, numerical results show that a commercial personal computer can execute the threshold approach based proposed clustering around $100$ milliseconds for $K=N=250$. 
\subsection{UAV Deployment: An Overall Orchestration}
The UAV location is the most conflicting variable that has a significant impact on the system performance due to its direct relationships with path loss, channel gains, LoS probability of signal and interference link budgets. Even for a given clustering and resource allocation setting, the deployment problem is a highly non-convex problem. From the clustering and resource allocation point of view, a change in the UAV location is seen as a change in environmental parameters. In this regard, the UAV deployment problem can be solved by meta-heuristic methods that runs a global search of UAV locations and evaluate the location fitness by proposed clustering and resource allocation procedure. There exist powerful meta-heuristic methods such as simulated annealing, particle swarm optimization, and genetic algorithm. Even though we investigate these three methods in the numerical results section, we only provide the algorithmic implementation details of the simulated annealing approach in Algorithm \ref{alg:SP123}.   
\section{Numerical Results}
\label{sec:res}
In this section, we present numerical results to validate the analytical expressions and disclose impacts of different system parameters and scenarios on the proposed system model. The default system parameters, if it is not stated otherwise, are listed in Table \ref{tab:sys_parameters}. In order to provide a deeper insight into the entangled relations between resource allocation and UAV deployment, we first focus on a single cluster performance based on the network setup shown in Fig. \ref{fig:anal_setup}.
\begin{table}[t]
	\centering
	\caption{Default system parameters.}
	\label{tab:sys_parameters}
	\resizebox{1\textwidth}{!}{
\begin{tabular}{|c|c|c|c|c|c|c|c|}
	\hline
	Parameter   & Value   & Parameter                                             & Value         & Parameter                                       & Value              & Parameter  & Value    \\ \hline \hline
	$W$     & 180 kHz & $\{a_i, b_i\}$                                    & $\{7, 0.2\}$  & $\{P_s, P_p, P_r\}$                         & \{46, 46, 30\} dBm & $m = x_i^j = y_i^k = z_p^j $ & $2$        \\ \hline
	$R$     & 500 m   & $\{a_n$, $b_n\}$                                       & \{13, 0.22\}       & $\{\eta^{\text{LoS}}, \eta^{\text{NLoS}}\}$ & \{1.6, 20\} dB     & $d$    & 0-1000 m \\ \hline
	$R^{'}$ & 100 m   & $\{H_p, H_s, H_n\}$ & $\{20,20,0\}$ & $f_c$                                       & 1.8 GHz            & $H_r$  & 0-1000 m  \\ \hline
\end{tabular}
}
\end{table}
 \begin{figure}[t]
 \centering
 \includegraphics[width=0.7\linewidth]{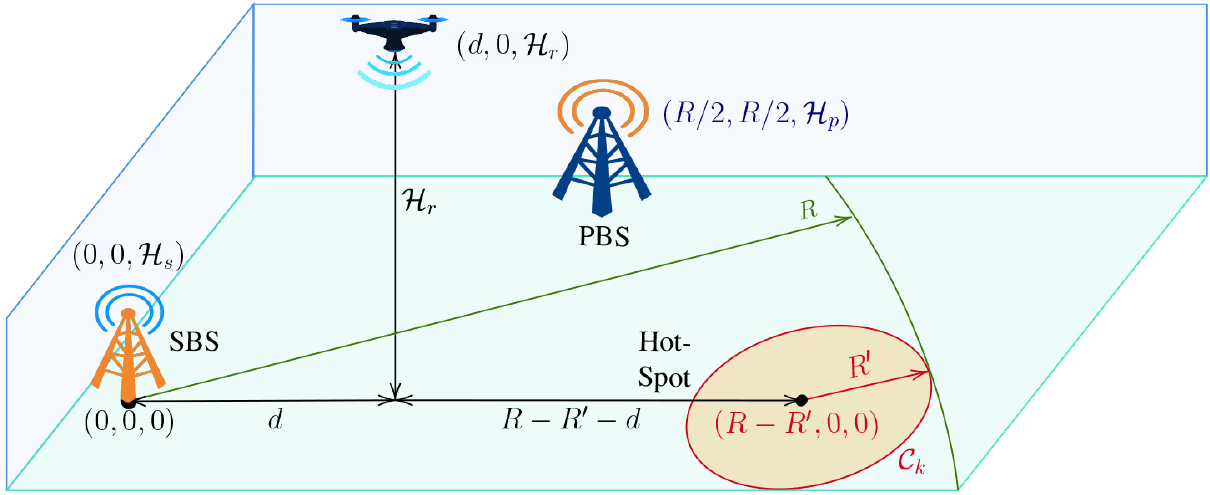}
\caption{Network setup for verification of derivations.}
\label{fig:anal_setup}
\end{figure}

\subsection{Validation of the Closed-Form Power Allocations}

\begin{figure}[t]
	\begin{subfigure}{.5\textwidth}
		\centering
		\includegraphics[width=1\linewidth]{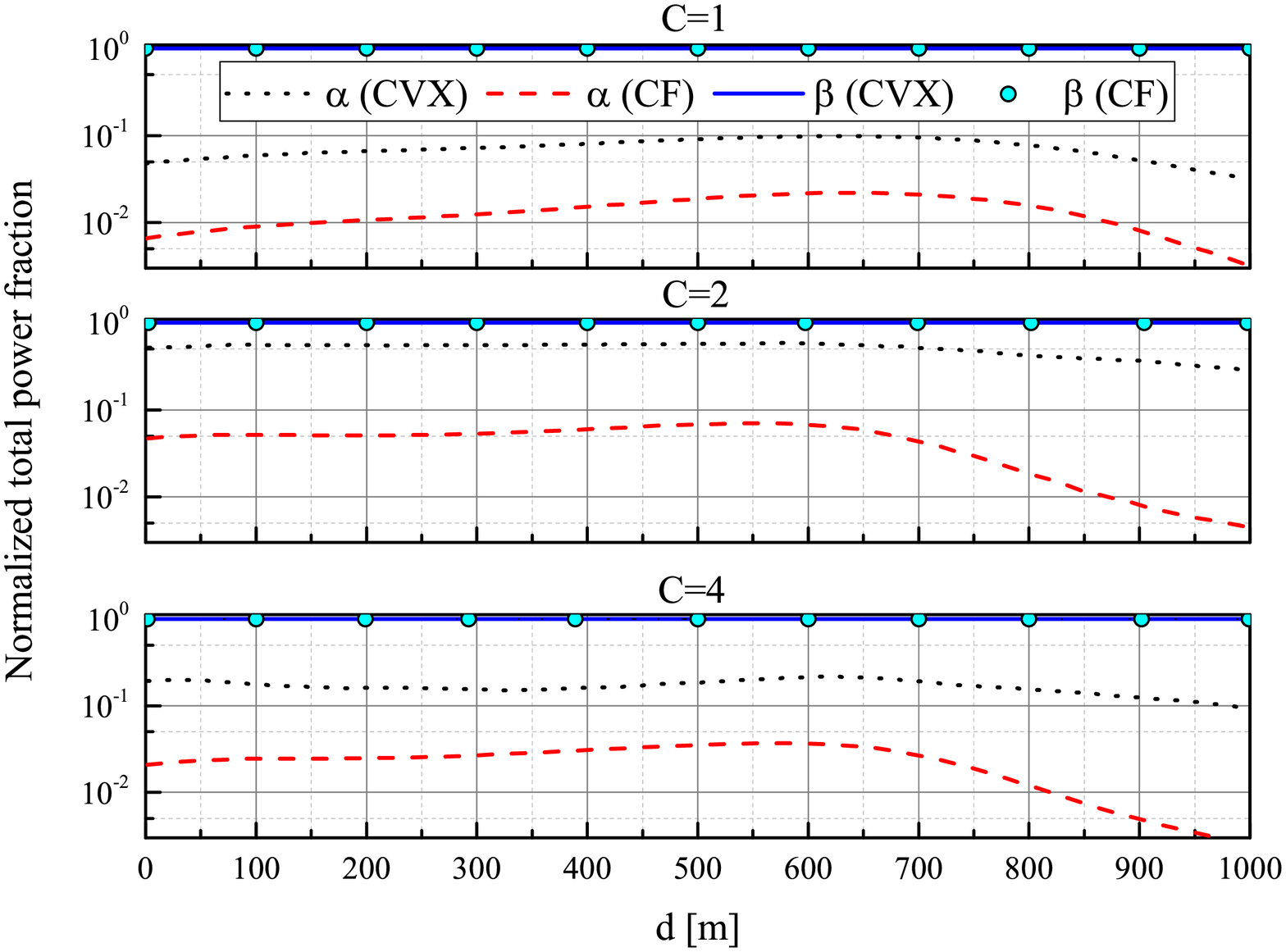}  
		\caption{}
		\label{fig:alpha_beta_distance}
	\end{subfigure}
	\begin{subfigure}{.5\textwidth}
		\centering
		\includegraphics[width=1\linewidth]{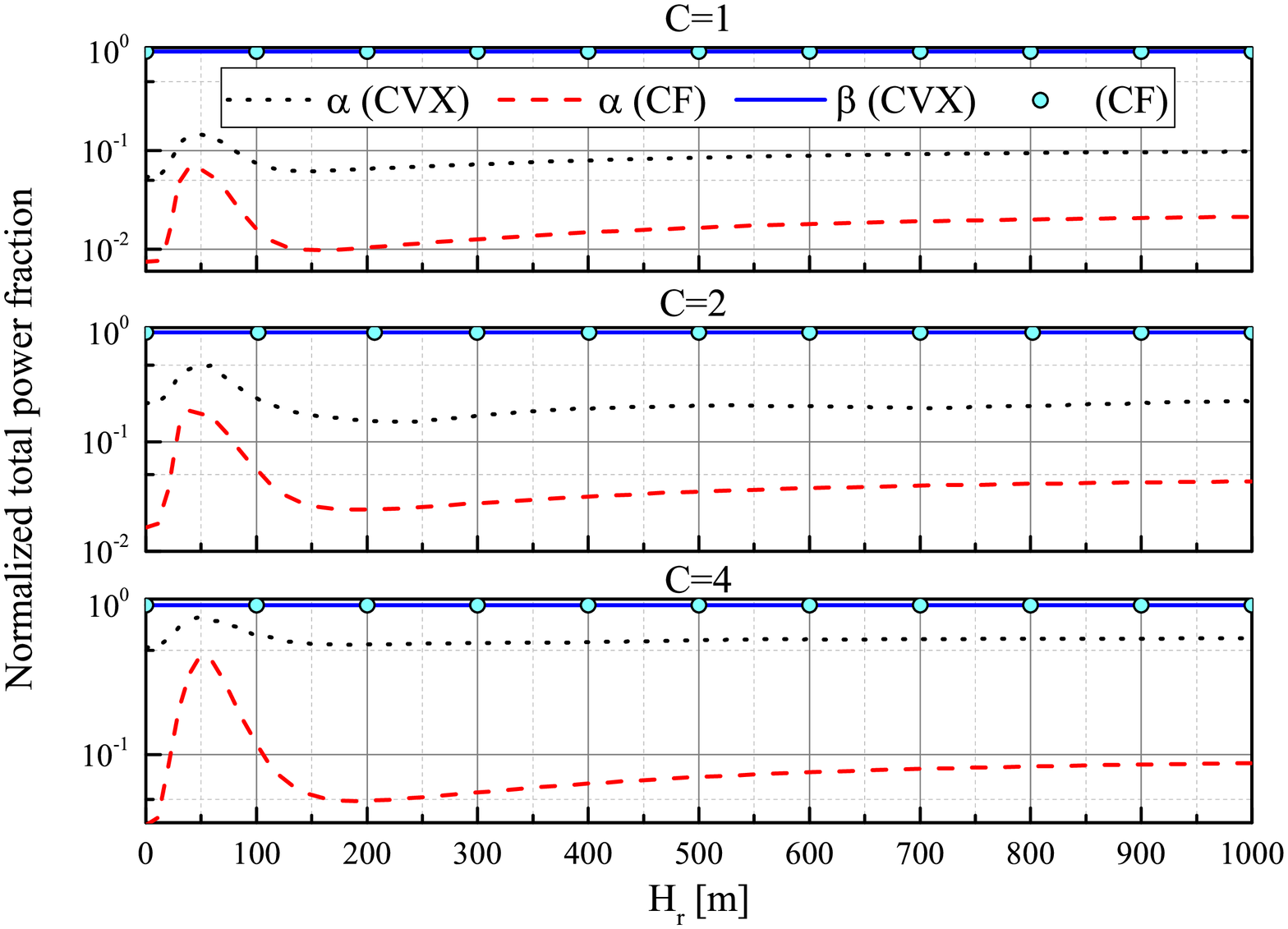}  
		\caption{}
		\label{fig:alpha_beta_height}
	\end{subfigure}
	\caption{Optimal $\alpha$ and $\beta$: a) for various distance  with the fixed $H_r = 250$ m; b) for various height with the fixed $d = R/2$ m.}
	\label{fig:Normalized_PAs}
\end{figure}

In Fig. \ref{fig:Normalized_PAs}, we show optimal power fractions for max-min throughput found using geometric programming (CVX~\cite{cvx}) in \eqref{CVX} and derived analytically in a closed-form (CF) in Lemma \ref{lem:SP1}. We consider four secondary NOMA users with the following coordinates: $\{ d^1_{\text{UAV}}; d^2_{\text{UAV}}; d^3_{\text{UAV}}; d^4_{\text{UAV}} \}$ = $\{(500,400,0); (450,400,0); (350,400,0); (300,400,0)\}$.  Fig. \ref{fig:alpha_beta_distance} illustrates the results for normalized optimal $\alpha$ and $\beta$ considering various $d$ locations of the UAV and different cluster sizes, $C$. For all cluster sizes and $d$ regions, it is noticed that the values of $\beta$ for both CVX and CF are equal to $1$. This is due to the fact that the access link has lower transmission power comparing to the backhaul link and the access link always provides minimum throughput for the considered system setup. On the other hand, the values of $\alpha$ are less than $1$ as the SBS decreases its transmission power to adopt the throughput of the backhaul link to that of the access link. Furthermore, it worth mentioning that the values of $\alpha$ in the CF are much lower than those in the CVX. It means that the CF solution provides a much power-efficient solution than the CVX. For example, for $C=1$ and $d =500$ m, $\alpha$ for the CF equals to $0.02$, while $\alpha$ for the CVX is $0.09$, which means that the CF solution provides more than $4$ times efficient power usage than the CVX method. Moreover, the CF method provides better performance in terms of computation time than the CVX as the CF approach does not use iteration in finding the optimal SINDR. For example, when $C=2$, the elapsed time for CF solution to find the optimal power allocation factors is $0.05$ sec, while the CVX method spends $230$ sec for the same purpose. Fig. \ref{fig:alpha_beta_height} plots optimal power fractions for different height ($H_r$) of the UAV. It is observed that in some $H_r$ regions $\alpha$ for both CVX and CF show a high increase comparing with other height regions. This happens due to the spatial expectation of the path loss and the LoS shown in \eqref{eq:spatial_PL} and \eqref{LoS_i_r}, respectively. The probability of LoS or NLoS is strongly dependent on the elevation angle. For example, at $10-80$ m height when the elevation angle between the SBS and the UAV is low, the effect of NLoS is stronger for the SBS-UAV link, as a result, the SBS needs to increase its transmission power to meet the end-to-end max-min throughput.
\begin{figure}[t]
	\begin{subfigure}{.5\textwidth}
		\centering
		\includegraphics[width=1\linewidth]{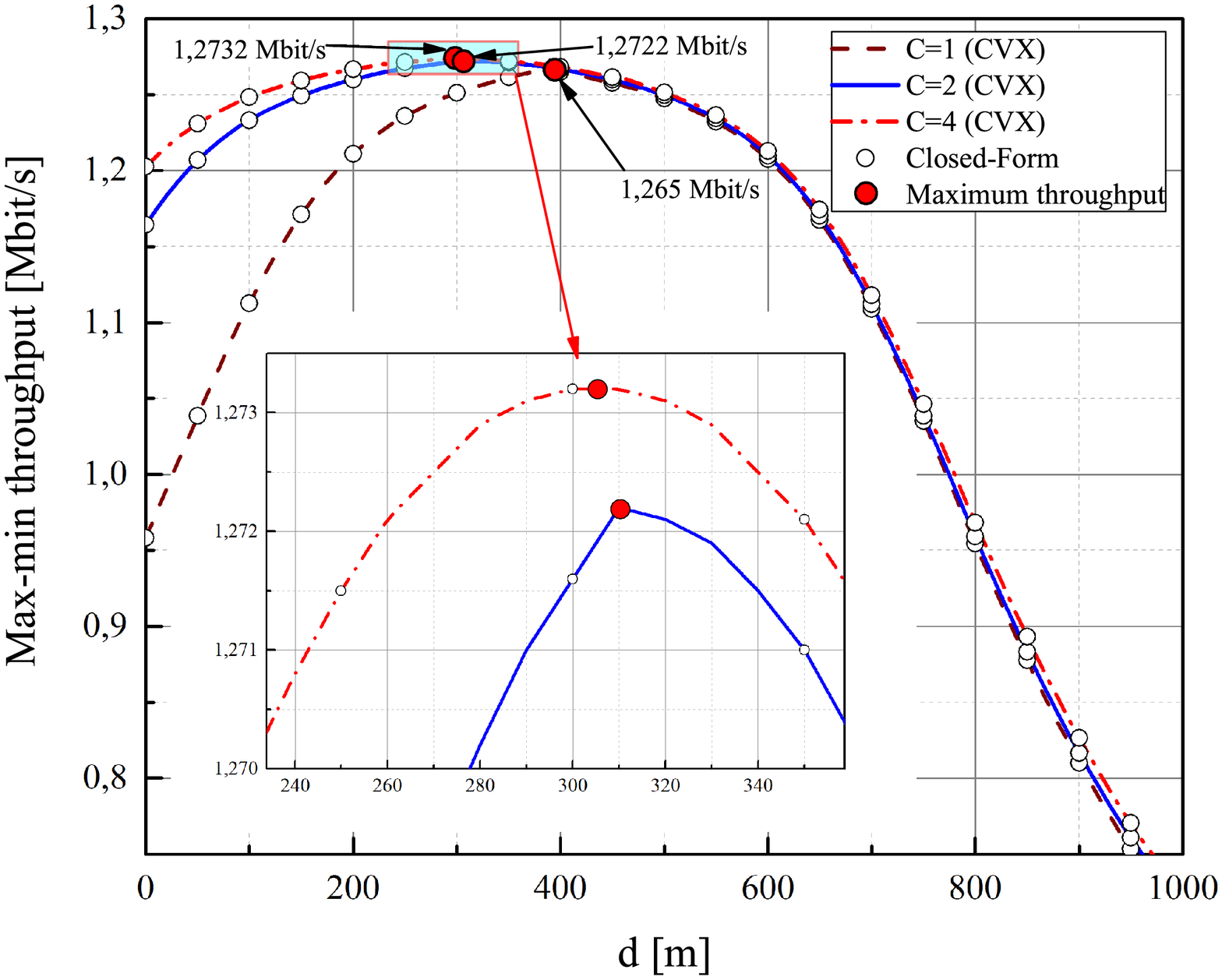}  
		\caption{For various distance with fixed $H_r = 250$ m.}
		\label{fig:optimal_gamma_distance}
	\end{subfigure}
	\begin{subfigure}{.5\textwidth}
		\centering
		\includegraphics[width=1\linewidth]{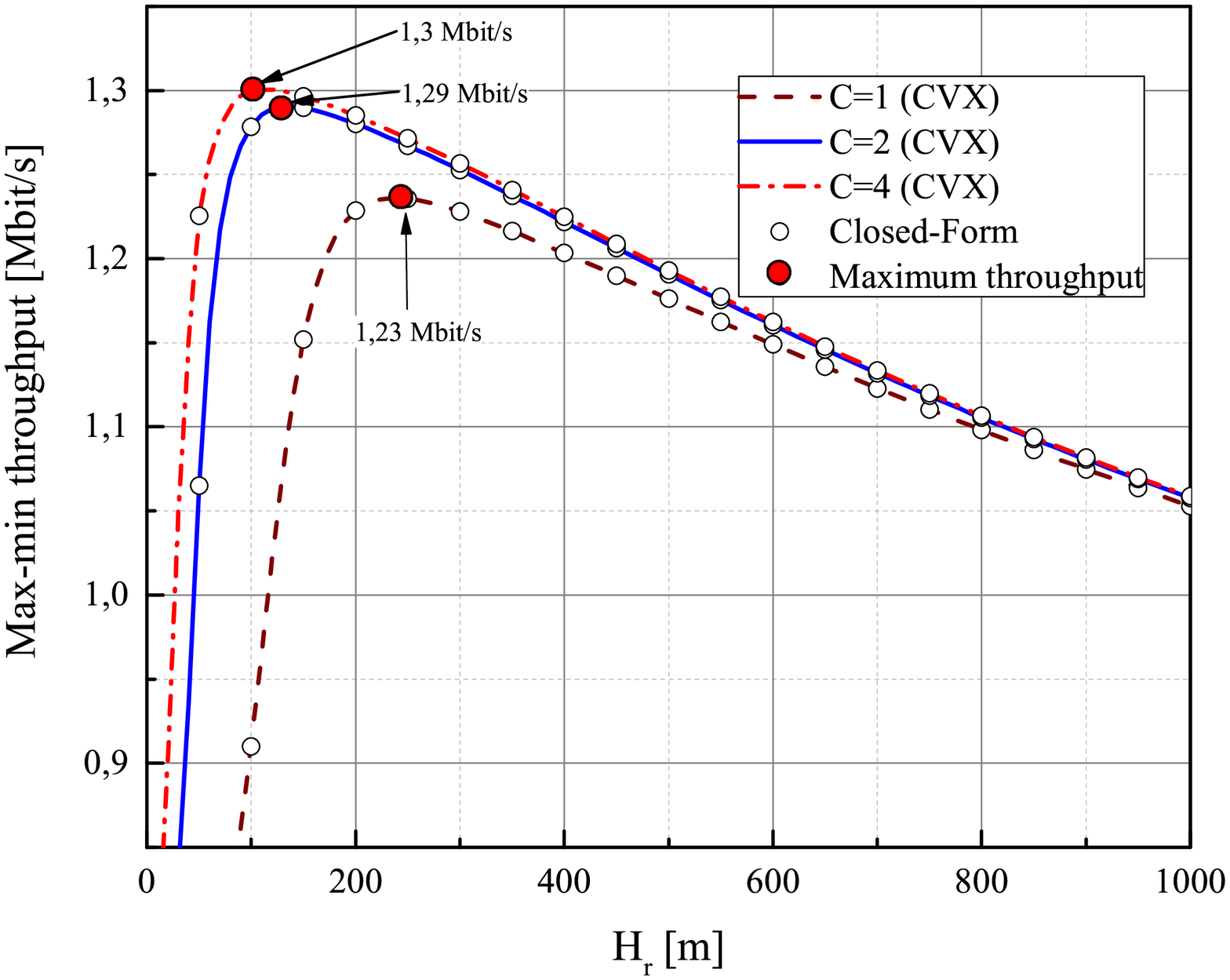}  
		\caption{For various height with fixed $d = R/2$ m.}
		\label{fig:optimal_gamma_height}
	\end{subfigure}
	\caption{Max-min throughput performance for clusters with different sizes.}
	\label{fig:max_min_thr_d_hr}
\end{figure}

Fig. \ref{fig:max_min_thr_d_hr} plots max-min throughputs for clusters with different sizes obtained using optimal power allocation factors shown in Fig. \ref{fig:Normalized_PAs}. It is noticed from the plot that throughputs obtained from the CVX and the CF are matched to each other. Fig. \ref{fig:optimal_gamma_distance} illustrates max-min throughput performance over various $d$ values. The figure shows that, when $C=1$, which is considered as an OMA system, the overall max-min throughput of the clusters equals to $1.265$ Mbit/s and $d = 400$ m, which is the center of the hot-spot. Furthermore, when the system model turns into NOMA mode with $C=2$ and $C=4$, the max-min throughput increases by showing $1.2722$ and $1.2732$ Mbit/s, respectively. In Fig. \ref{fig:optimal_gamma_height}, the max-min throughput is plotted versus different height of the UAV. Similarly, as in the previous sub-plot, the OMA mode obtains the lowest throughput with $1.23$ Mbit/s. The NOMA mode with $C=2$ obtains $1.29$ Mbit/s, while $C=4$ gains the throughput of $1.3$ Mbit/s. It is obvious that the throughput improvement from  $C=1$ to $C=2$ is $0.06$ Mbit/s, whereas that from $C=2$ to $C=4$ is only $0.01$ Mbit/s. This can be explained by the fact that the rise of the cluster size increases interference from NOMA users inside the cluster, which affects negatively on the achievable throughput. Thus, the higher cluster size results in less throughput improvement.

\begin{figure}[t]
	\begin{subfigure}{.5\textwidth}
		\centering
		\includegraphics[width=1\linewidth]{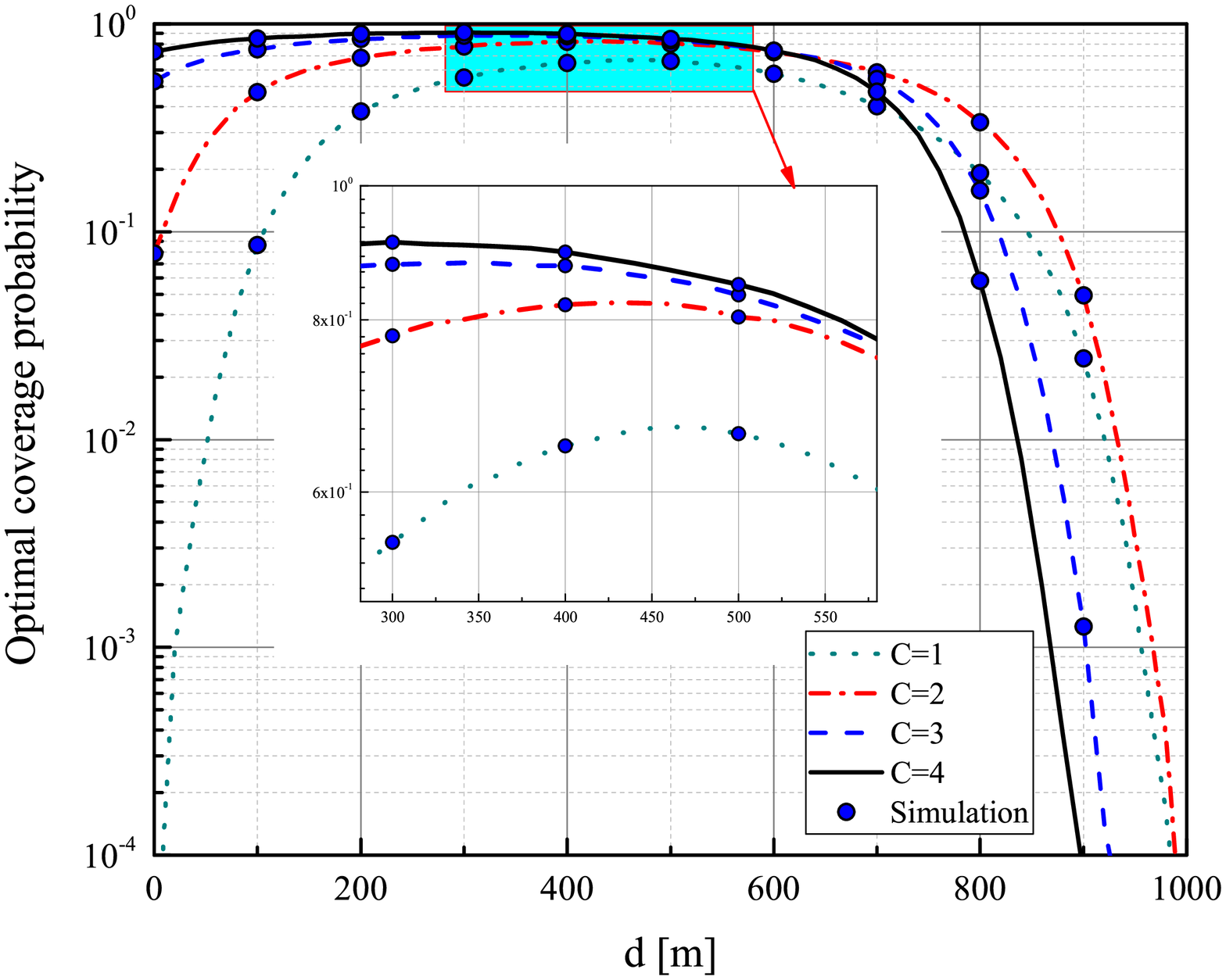}  
		\caption{For various distance with fixed $H_r = 250$ m.}
		\label{fig:optimal_cov_prob_d}
	\end{subfigure}
	\begin{subfigure}{.5\textwidth}
		\centering
		\includegraphics[width=1\linewidth]{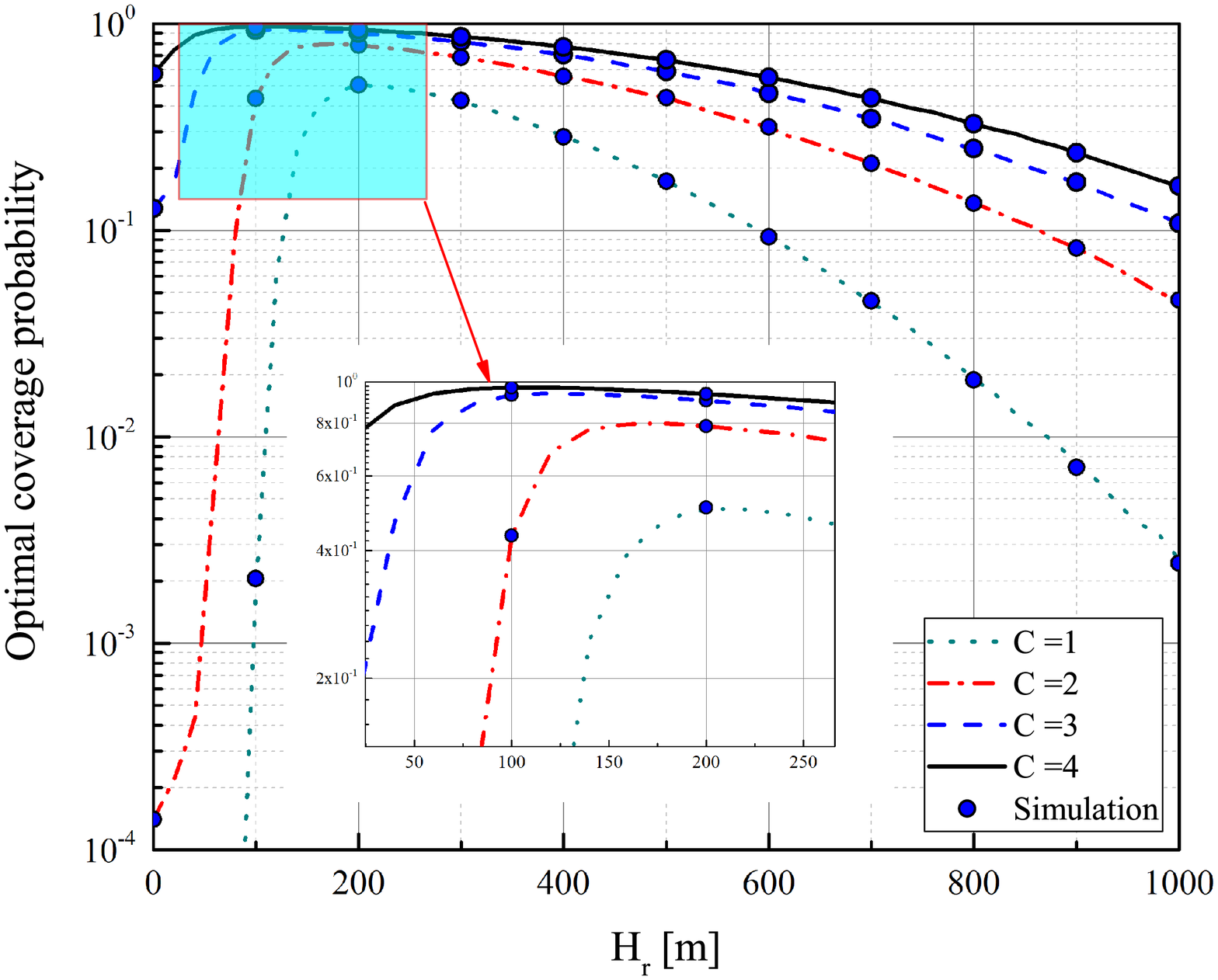}  
		\caption{For various height with fixed $d = R/2$ m.}
		\label{fig:optimal_cov_prob_Hr}
	\end{subfigure}
	\caption{Optimal coverage probability for clusters with different sizes.}
	\label{fig:Optimal coverage probability}
\end{figure}

\subsection{Validation of the Coverage Performance Analysis}

Fig. \ref{fig:Optimal coverage probability} demonstrates optimal coverage probability for various $d$ and $H_r$ values considering different cluster sizes with $\bar{R} = 1.3$ Mbit/s. It can be seen from Fig. \ref{fig:optimal_cov_prob_d} that maximum optimal coverage probability for all clusters are at about  the center of the hot-spot. The overall observation is that the coverage performance of $C=4$ outperforms that of thr other clusters. However, after $700$ m, we can see that the performance of $C=4$ degrades comparing to the other clusters. This phenomenon can be explained as follows. As it was explained in Fig. \ref{fig:Normalized_PAs}, the secondary NOMA users are located on the $x$-axis within hotspot radius, where $d^1_{\text{UAV}}$ is the weakest channel and $d^4_{\text{UAV}}$ is the strongest one. When the location of the UAV is on the right side of the hot-spot circle, SU$_1$ becomes the strongest user and SU$_4$ will be the weakest one. However, regarding the system setup, the stronger user, i.e., SU$_1$, has higher power allocation fractions and the weaker users have lower allocated powers due to this reason the optimal coverage performance of higher cluster sizes degrades for the considered UAV locations. In Fig. \ref{fig:optimal_cov_prob_Hr}, we show how the location of the UAV in different heights with a given $d = R/2$ can impact on the coverage performance of the secondary users. As it is expected, $C=4$ obtains the best coverage performance comparing to other cluster sizes by achieving the maximum coverage at $120$ m. Moreover, it is noticed that when the cluster size is lower the UAV height for obtaining the optimal coverage is higher. For example, when $C=1$, the UAV height for the maximum coverage is $200$ m.

\begin{figure}[t]
	\begin{subfigure}{.5\textwidth}
		\centering
		\includegraphics[width=1\linewidth]{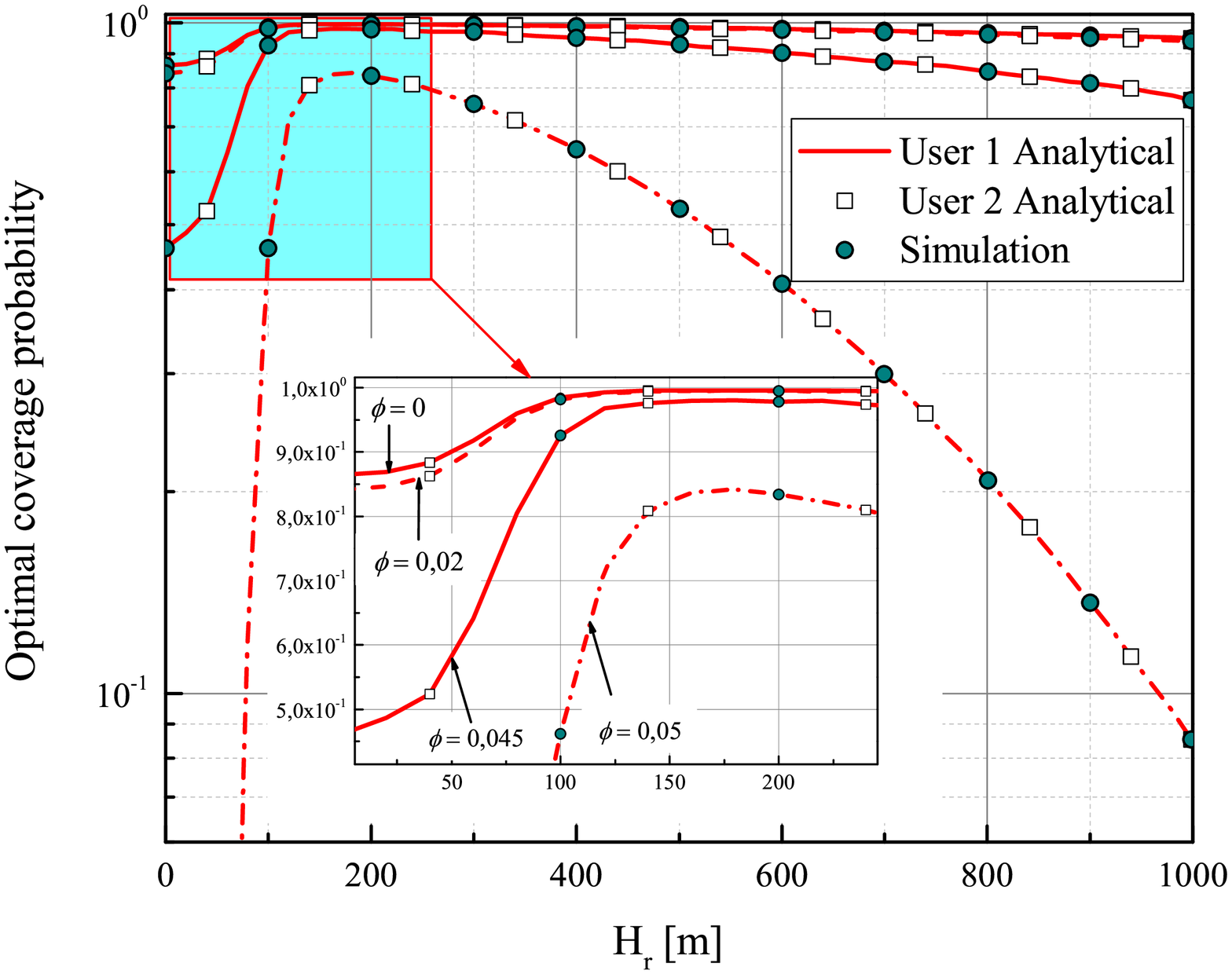}  
		\caption{For various $H_r$ and different HI conditions.}
		\label{fig:HIs}
	\end{subfigure}
	\begin{subfigure}{.5\textwidth}
		\centering
		\includegraphics[width=1\linewidth]{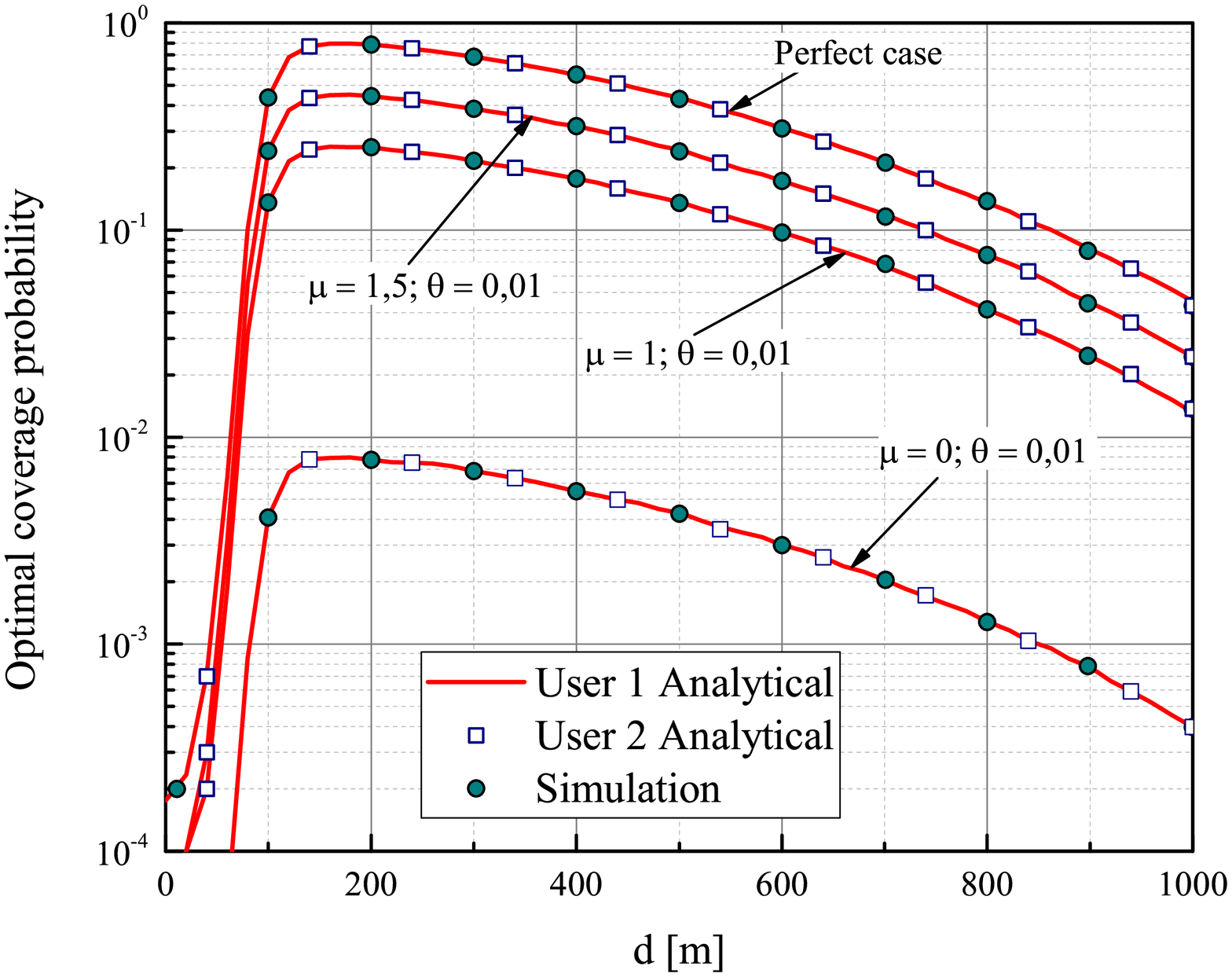}  
		\caption{for various $d$ and different CSI conditions.}
		\label{fig:opt_cov_imp_csi}
	\end{subfigure}
	\caption{Optimal coverage probability considering system imperfections for $C=2$.}
\end{figure}

\begin{figure}[t]
    \centering
	\includegraphics[width=0.5\textwidth]{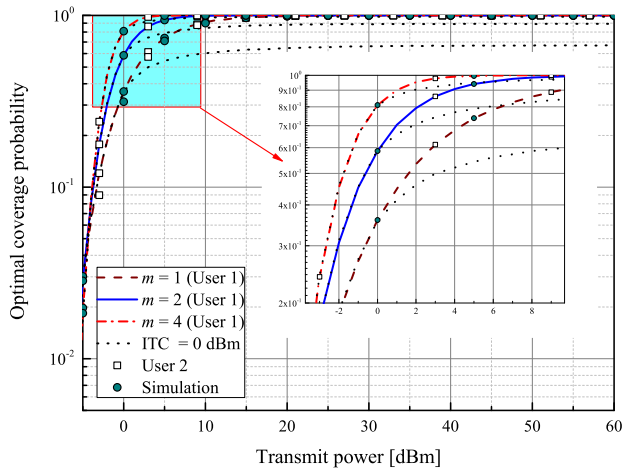}
	\caption{Optimal coverage probability for various transmission power with $C=2$ and different $m$ parameters.}
	\label{fig:impSIC_ITC}
\end{figure}

Fig. \ref{fig:HIs} aims to analyze the effect of the HI level on the optimal coverage probability for a given $C=2$, with coordinates $\{ d^1_{\text{UAV}}; d^2_{\text{UAV}}\}$ = $\{(500,400,0); (450,400,0)\}$, $\phi = \{0, 0.02, 0.045, 0.05\}$ and $\bar{R} = 0.8$ Mbit/s. It is worth mentioning that we use the CF-based optimal power allocation fractions for both NOMA users. Hence, NOMA users obtain the same coverage performance for all considered UAV height. Noticeably, the coverage probability degrades as the level of HI increases since HIs add extra interference level on the system. Having said that, the coverage degradation from an ideal hardware case to $\phi = 0.02$ case is $2.5 \%$ and $\phi = 0.045$ case is $40 \%$, respectively, at $40$ m. In addition, when $\phi = 0.05$, the users obtain the worst coverage performance for all UAV locations. Moreover, in this and forthcoming figures, the Monte Carlo simulations perfectly coincide with the analytical results by validating the accuracy of analytical derivations.

Fig. \ref{fig:opt_cov_imp_csi} demonstrates the impact of the imperfect CSI on the coverage probability over various $d$ locations by considering  other system imperfections in ideal setup, i.e., $\phi=0$ and $\backepsilon=0$. It is seen that NOMA users obtain the best coverage performance with perfect CSI model by achieving $0.8$ coverage probability at $180$ m. Furthermore, when we set $\theta=0.01$, the best coverage performance of users degrades to $0.008$. However, when the value of $\mu$ increases, the coverage performance improves getting closer to the perfect mode performance. For example, when $\mu=1$ and $\mu=1.5$, users receive the maximum coverage probability of $0.25$ and $0.45$, respectively. 

In Fig. \ref{fig:impSIC_ITC}, we demonstrate the impact of the small scale fading parameter $m$ on the optimal coverage probability considering $\bar{R} = 0.4$ Mbit/s. It is obvious from the plot that the increase of the parameter $m$ provides better coverage performance. This happens because the parameter $m$ represents the number of multi-path
components of the channel. When $m$ is higher, the number of multi-path increases providing diversity. For example, at $4$ dBm transmission power, the system setup with $m=4$ obtains the full coverage, while $m=1$ and $m=2$ achieves the coverage probability of $0.7$ and $0.9$, respectively.
Moreover, we also can observe the impact of the ITC on the coverage probability. We set ITC = $0$ dBm and notice that the coverage performance for all $m$ values saturate after a certain  transmission power and do not achieve the full coverage performance.

\subsection{Performance Evaluation of the User Clustering Approach}
To evaluate the performance of the proposed user clustering and channel assignment approach, we compare it with an optimal ILP benchmark, which is computed by Gurobi and MOSEK solvers of the CVX \cite{cvx}. The elapsed times for both approaches are shown in Table \ref{tab:tlapse_comp}, where the number of channels and users are kept the same for simplicity (i.e., $K=N$). It is obvious from Table \ref{tab:tlapse_comp} that the time complexity difference between the two approaches significantly increases as $K=N$ increases. For instance, the ILP-CVX takes more than 3, 4, and 5 orders magnitude of time for $K=N$ is 20, 25, and 30, respectively. At this point, we must note that the proposed approach reach 100\% accuracy at all cases. Since the ILP-CVX takes extremely long times for $K=N>30$, we show the elapsed time only for the proposed approach in Table \ref{tab:tlapse_GBAP}, which shows that LBA approach can provide a solution in less than half a second for 500 users.

\begin{table}[t]
\caption{Elapsed time comparison between the ILP benchmark and LBA approach.}
\label{tab:tlapse_comp}
\centering 
\resizebox{0.65\textwidth}{!}{%
\begin{tabular}{l|l|c|c|c|c|c|c|}
\cline{2-8}
 & \textbf{K=N}     & 5      & 10     & 15     & 20     & 25     & 30     \\ \hline
\multicolumn{1}{|l|}{\multirow{2}{*}{\textbf{Time [s]}}} & \textbf{LBA}    & .00148 & .00172 & .00268 & .00375 & .00879 & .01294 \\ \cline{2-8} 
\multicolumn{1}{|l|}{}                               & \textbf{ILP-CVX} & 3.5    & 3.8    & 6.6    & 17.9   & 292.5  & 3185.9 \\ \hline
\end{tabular}%
}
\end{table}

\begin{table}[t]
\caption{Elapsed time for the proposed clustering approach.}
\label{tab:tlapse_GBAP}
\centering 
\resizebox{0.7\textwidth}{!}{%
\begin{tabular}{|l|c|c|c|c|c|c|c|c|c|c|c|c|c|c|c|c|c|c|c|c|}
\hline
\textbf{K=N}          & 50   & 100  & 150  & 200  & 250  & 300  & 350  & 400  & 450  & 500    \\ \hline
\textbf{Time [s]} & .019 & .033 & .039 & .064 & .109 & .171 & .219 & .283 & .354 & .470 \\ \hline
\end{tabular}%
}
\end{table}
\begin{figure}[t]
	\begin{subfigure}{0.5\textwidth}
		\centering
	\includegraphics[width=1\linewidth]{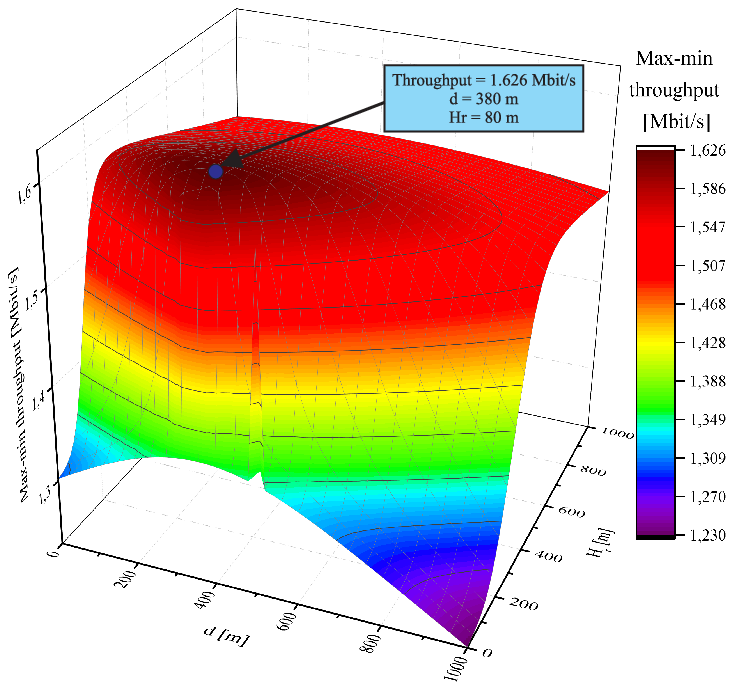}
	\caption{}
	\label{fig:optimal_throughput_for_2users}
	\end{subfigure}
	\begin{subfigure}{0.5\textwidth}
		\centering
    \includegraphics[width=1\linewidth]{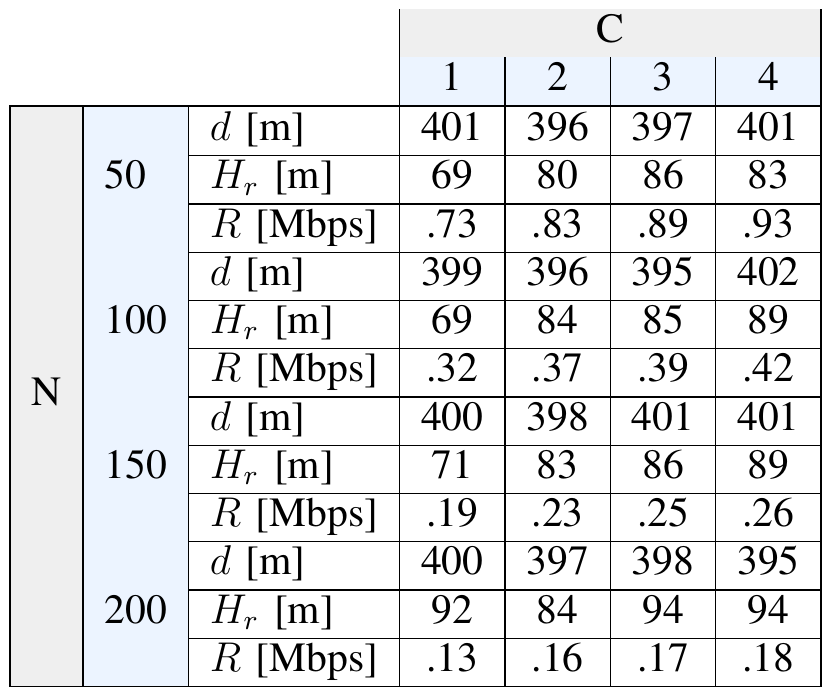}
    \caption{}
    \label{fig:opt_UAV_loc}
	\end{subfigure}
	\caption{Impact of the UAV location and cluster size on the max-min throughput: a) Max-min throughput for various height and distance when $C=2$ and b) Optimal UAV location and max-min throughput for clusters with different sizes (simulated annealing algorithm).}
\end{figure}

\subsection{Impacts of User Density and Cluster Size on Optimal UAV Deployment }
In Fig.  \ref{fig:optimal_throughput_for_2users}, we find an optimal UAV location for cluster size $C=2$, where we vary distances of $d$ and $H_r$. As it is shown from the plot, the max-min throughput of $1.626$ Mbit/s is obtained at optimal locations when $d = 380$ m and $H_r = 80$ m. It is noticed that the max-min throughput is lower when $H_r$ is closer to the ground, i.e., $0-10$ m. This is due to the higher probability of NLoS link as the UAV is located below the height of the SBS. On the other hand, when $H_r>200$ m, the achieved max-min throughput degrades due to the increase of the path-loss impact on the system performance when the UAV flies at a higher altitude. Moreover, the optimal location of $d = 380$ m is close to the center of the hot-spot. The reason for that is as follows. The UAV has lower available transmission power than the SBS. Therefore, the UAV needs to fly closer to NOMA users in order to provide maximum throughput in the access link.  
 
Fig. \ref{fig:opt_UAV_loc}  illustrates users clustering and searching the max-min throughput as well as the optimal UAV location, i.e., $d$ and $H_r$, described in Algorithm \ref{alg:SP123}.  Here, we consider cases when the total number of secondary NOMA users, which are randomly distributed within the hot-spot radius, is equal to $N = 50,100,150,200$ and the maximum cluster size is $C=4$. Moreover, we provide the results of the simulated annealing algorithm approach as it showed $4$ times quicker elapsed time performance comparing with the particle swarm optimization. When $N=50$ and $C=1$, the total number of clusters is $50$ and each cluster has just $1$ user, which can be considered as an OMA case. The max-min throughput for all clusters is  $0.73$ Mbit/s. Then, when $C=2$, the max-min throughput increases by $0.1$ Mbit/s. This shows the outperformance of the NOMA mode comparing to the OMA one. Furthermore, when the cluster size increased to $C=3$, the max-min throughput is $0.89$ Mbit/s, which is for $0.06$ Mbit/s higher than that of $C=2$ mode. When $C=4$, the max-min throughput equals to $0.93$ Mbit/s, which $0.04$ Mbit/s higher than the previous cluster size. As can be seen from the pattern of the throughput increase, considerable performance improvement happens when the OMA mode becomes NOMA with a cluster size of $2$. On the other hand, the throughput performance shows only a slight improvement when the cluster size is increased from $C=2$ to $C=3$ and from $C=3$ to $C=4$. The reason for that is the next: despite the broad bandwidth of higher clusters comparing to lower clusters, the increase of the number of users in each cluster also raises the level of interference within that cluster. As a result, that interference negatively impacts on the end-to-end SINDR of users. Furthermore, when $N=100$, we see that the max-min throughput for $C=1$ and $C=2,3,4$ decreases for $0.4$ and $0.5$ Mbit/s, respectively, comparing to the same cluster sizes when $N=50$. Similarly, the max-min throughput reduces by further increasing the number of users to $150$ and $200$. The throughput reduction happens since the total transmission power is divided among the NOMA users, which means that each user achieve less power for the signal detection. It is worth noting that the optimal $d$ and $H_r$ after averaging values of all clusters are equal to $399$ m and $83$ m, respectively.
\section{Conclusions}
\label{sec:conc}
The CCR-NOMA is an inherent remedy to achieve a high spectral efficiency at large-scale networks. Considering expeditious $3$D deployment capabilities of UAVs, their use as a cooperative relay paves the way for enhanced network performance. In this regard, this paper investigated the optimal UAV deployment by accounting for clustering, channel assignment, and resource allocation sub-problems. To reduce the computational time complexity, we derived closed-form solutions for optimal resource allocations and coverage probabilities for a user cluster. The closed-form solutions are then used by the proposed user clustering and channel assignment heuristics, which is fast yet highly accurate. The UAV deployment problem exploits this framework to evaluate the fitness of locations to find the optimal UAV placement.

\textbf{}\appendices
\numberwithin{equation}{section}
\section{Proof of Proposition 1}
\label{Appendix 1}
Considering the ITC imposed by PU$_k$, the CDF $F_{\gamma_{k,n}^{i,j}}(\bar{\gamma}_k^i)$  in \eqref{coverage_prob_r} can be further written as
\begin{align}
\label{CDFsr}
F_{\gamma_{k,n}^{i,j}}(\bar{\gamma}_k^i) =& \text{Pr}\left[ \frac{X_i^j v^n_k}{X_i^j I_{k,n}^i + X_i^j \sigma^2_{\phi^j_i} + E^j_i + Z_p^j I^j_p + \bar{\sigma}^2_j }< \bar{\gamma}_k^i, P_i^k < \frac{\text{ITC}_k}{Y_i^k} \right] \nonumber\\
&+\text{Pr}\left[ \frac{ X_i^j v^n_k}{ X_i^j I_{k,n}^i +  X_i^j \sigma^2_{\phi^j_i} + E^j_i  + Z_p^j I^j_{\bar{p}} +   \bar{\sigma}^2_{\bar{j}} }< \bar{\gamma}_k^i, P_i^k > \frac{\text{ITC}_k}{Y_i^k} \right] \nonumber\\
=&\underbrace{ \text{Pr}\left[X_i^j < Z_p^j \mathcal{I}_i  + \mathcal{S}_i + \mathcal{E}_i  , Y_i^k < \Lambda_i \right]}_{\text{$\Delta$}} + \underbrace{\text{Pr}\left[X_i^j < Z_p^j Y_i^k \mathcal{U}_i + Y_i^k \mathcal{V}_i + \mathcal{E}_i, Y_i^k > \Lambda_i \right]}_{\text{$\Upsilon$}}.
\end{align} 
We can further rewrite the term $\Delta$ in \eqref{CDFsr} as follows
{\allowdisplaybreaks
\begin{align}
\label{CDFsr_int}
\Delta =& \int_{0}^{\infty}  f_{Z_p^j}(z) \int_{0}^{z_p^j \mathcal{I}_i + \mathcal{E}_i  + \mathcal{S}_i} \hspace{-1cm}f_{X_i^j}(x){\text d}x {\text d}z  \int_{0}^{\Lambda_i}  f_{Y_i^k}(y){\text d}y 
= \frac{\gamma_{inc}\left( {y_i^k}, {y_i^k} \Lambda_i \right)}{\Gamma({y_i^k})}  \int_{0}^{\infty}  \frac{ ({z_p^j})^{z_p^j} z^{{z_p^j}-1} \exp \left[-{z_p^j} z \right]}{ \Gamma({z_p^j})} \nonumber\\
&\times \frac{\gamma_{inc}\left( {x_i^j}, {x_i^j} \left(z \mathcal{I}_i + \mathcal{E}_i  + \mathcal{S}_i\right) \right)}{\Gamma({x_i^j})} {\text d}z 
\stackrel{(a)}{=}  \frac{\gamma_{inc}\left( {y_i^k}, {y_i^k} \Lambda_i \right)}{\Gamma({y_i^k})} \left( 1 - \frac{  \exp \left[-{x_i^j} \left(\mathcal{E}_i  + \mathcal{S}_i \right) \right] }{({z_p^j})^{-{z_p^j}}~ \Gamma({z_p^j})} \right) \nonumber\\
&  \times \sum_{q=0}^{{x_i^j}-1}
\frac{ \left(z \mathcal{I}_i + \mathcal{E}_i  + \mathcal{S}_i \right)^q  }{{x_i^j}^{-q} q!}        \int_{0}^{\infty}  z^{{z_p^j} - 1} {\text d}z
	 	  \exp \left[-z \left({z_p^j} + {x_i^j} \mathcal{I}_i \right) \right]
\stackrel{(b)}{=} \frac{\gamma_{inc}\left( {y_i^k}, {y_i^k} \Lambda_i \right)}{\Gamma({y_i^k})} \nonumber \\ 
& - \frac{\gamma_{inc}\left( {y_i^k}, {y_i^k} \Lambda_i \right)}{\Gamma({y_i^k})} \frac{ ({z_p^j})^{{z_p^j}} \exp \left[-{x_i^j} \left(\mathcal{E}_i  + \mathcal{S}_i \right) \right] }{\Gamma({z_p^j})}  \sum_{q=0}^{{x_i^j}-1}
	\frac{({x_i^j})^q}{q!} \sum_{l=0}^{q} \binom{q}{l} \frac{ \left( \mathcal{E}_i  + \mathcal{S}_i \right)^q (\mathcal{I}_i)^{l} \Gamma(z_p^j + l)}{\left(z_p^j + x_i^j \mathcal{I}_i  \right)^{z_p^j+l}}, 
\end{align}
}where $(a)$ follows from using the series representation of the lower incomplete Gamma function given as $
\gamma_{inc}(m,\xi) = \Gamma(m) - \Gamma(m)\exp\left(-\xi\right) \sum_{i=0}^{m-1} \frac{\xi^i}{i!}$ for $m \in \mathbb{Z}^{+}$ and $(b)$ follows from the binomial series expansion. 
Then, we extend the term $\Upsilon$ in \eqref{CDFsr} as 
\begin{align}
\label{V}
\Upsilon&= \int_{0}^{\infty} f_{Z_p^j}(z) \underbrace{\int_{\Lambda_i}^{\infty} \int_{0}^{yz \mathcal{U}_i  + y \mathcal{V}_i + \mathcal{E}_i}  f_{X_i^j}(x) f_{Y_i^k}(y){\text d}x {\text d}y}_{\Upsilon_1}  {\text d}z,
\end{align}
where the term $\Upsilon_1$ can be further expanded by using the series representation of the lower incomplete Gamma function and the binomial series expansion and written as
\begin{align}
\label{V1}
\Upsilon_1 =& \frac{\Gamma\left( {y_i^k}, {y_i^k} \Lambda_i \right) }{\Gamma({y_i^k})} - \frac{({y_i^k})^{{y_i^k}} \exp\left[-{x_i^j} \mathcal{E}_i \right] }{\Gamma({y_i^k})}    \sum_{\Bbbk=0}^{{x_i^j}-1} \frac{({x_i^j})^\Bbbk }{\Bbbk !} \sum_{\jmath=0}^{\Bbbk} \binom{\Bbbk}{\jmath} (\mathcal{E}_i)^{\Bbbk-\jmath} \left(z \mathcal{U}_i  +  \mathcal{V}_i \right)^\jmath \nonumber\\  
&\times \frac{\Gamma\left({y_i^k}+\jmath,~  \Lambda_i \left[ {y_i^k} + z {x_i^j} \mathcal{U}_i +  {x_i^j} \mathcal{V}_i \right] \right) }{\left[ {y_i^k} + z {x_i^j} \mathcal{U}_i  +  {x_i^j} \mathcal{V}_i \right]^{{y_i^k}+\jmath} },
\end{align}
Further, inserting \eqref{V1} into \eqref{V}, and using the series representation of the upper incomplete Gamma function of $\Gamma(m,c) = \Gamma(m)\exp\left(-c\right) \sum_{i=0}^{m-1} \frac{c^i}{i!}$ as well as the binomial series expansion, the term $\Upsilon$ can be rewritten as
\begin{align}
	\label{V_prefinal}
	\Upsilon =& \frac{\Gamma\left( {y_i^k}, {y_i^k} \Lambda_i \right) }{\Gamma({y_i^k})} - \frac{ ({z_p^j})^{{z_p^j}}}{\Gamma({z_p^j})  } \frac{({y_i^k})^{{y_i^k}} \exp\left[-{x_i^j} \mathcal{E}_i \right] }{\Gamma({y_i^k})}    \sum_{\Bbbk=0}^{{x_i^j}-1} \frac{({x_i^j})^\Bbbk }{\Bbbk !} \sum_{\jmath=0}^{\Bbbk} \binom{\Bbbk}{\jmath} \Gamma({y_i^k} +\jmath) (\mathcal{E}_i)^{\Bbbk-\jmath}   \nonumber\\
	&\times \sum_{t=0}^{p} \binom{p}{t}  \left({y_i^k}+{x_i^j} \mathcal{V}_i  \right)^{p-t} \left({x_i^j} \mathcal{U}_i \right)^t \exp\left[- \Lambda_i \left({y_i^k}+{x_i^j} \mathcal{V}_i \right)  \right]   \frac{\sum\limits_{u=0}^\jmath     (\mathcal{V}_i)^{\jmath-u} (\mathcal{U}_i)^u \sum\limits_{p=0}^{{y_i^k}+\jmath-1} \frac{(\Lambda_i)^p }{p !}}{\left({y_i^k}+{x_i^j} \mathcal{V}_i \right)^{{y_i^k}+\jmath}}   \nonumber\\
	& \times  \underbrace{\int_{0}^{\infty} \frac{z^{ {z_p^j}+u+t-1} \exp\left[- z \left({z_p^j} + {x_i^j} \Lambda_i \mathcal{U}_i   \right) \right]}{ \left(1+ \frac{{x_i^j} \mathcal{U}_i }{\left({y_i^k}+{x_i^j} \mathcal{V}_i \right) } z \right)^{{y_i^k}+\jmath} } {\text d}z}_{\Upsilon_2}, 
	\end{align}
Now, representing $\exp\left[-a z\right]$ and $(1+b z)^{-c} $ in terms of Meijer G-functions \cite[Eqs. (7.34.3.46.1) and (7.34.3.271.1)]{function_wolfram} respectively as $\MeijerG*{1}{0}{0}{1}{-}{0}{a z}$ and $\frac{1}{\Gamma(c)} \MeijerG*{1}{1}{1}{1}{1-c}{0}{b \gamma}$, we can reformulate the term $\Upsilon_2$ by
\begin{align}
\label{V_2}
\Upsilon_2 = \frac{1}{ \Gamma({y_i^k}+\jmath)} \int_{0}^{\infty} z^{{z_p^j}+u+t-1} ~\MeijerG*{1}{1}{1}{1}{1-{y_i^k}-\jmath}{0}{\frac{{x_i^j} \mathcal{U}_i z }{{y_i^k}+{x_i^j} \mathcal{V}_i  } } \MeijerG*{1}{0}{0}{1}{-}{0}{ \left({y_i^k}+{x_i^j} \mathcal{V}_i \right)  z} {\text d}z. 
\end{align}
Further, using \cite[Eq. (21)]{Adamchik} for $\Upsilon_2$ and after some mathematical manipulations, the term $\Upsilon$ can be expressed by 
\begin{align}
\label{V_final}
\Upsilon= & \frac{\Gamma\left( {y_i^k}, {y_i^k} \Lambda_i \right) }{\Gamma({y_i^k})} - \frac{ \left({z_p^j} \right)^{{z_p^j}}
 \exp\left[- \Lambda_i \left({y_i^k} + {x_i^j} \mathcal{V}_i  \right) \right]}{\Gamma({z_p^j}) }    \frac{\left({y_i^k} \right)^{{y_i^k}}
 \exp\left[-{x_i^j} \mathcal{E}_i\right] }{(\mathcal{U}_i)^{z_p^j} \, \Gamma({y_i^k})} \sum_{u=0}^{\jmath} \left( \mathcal{V}_i\right)^{\jmath-u}   \nonumber\\ 
 &\hspace{-0.3cm}\times \sum_{\Bbbk=0}^{{x_i^j}-1} \frac{\left
 ({x_i^j} \right)^{\Bbbk-{z_p^j} - u } }{\Bbbk !}  \sum_{\jmath=0}^{\Bbbk} \binom{\Bbbk}{\jmath} \left(\mathcal{E}_i\right)^{\Bbbk-\jmath} \sum_{p=0}^{{y_i^k} + \jmath-1}   \frac{\left(\Lambda_i \right)^p }{p !} \sum_{t=0}^{p} \binom{p}{t}     
	 \frac{\MeijerG*{2}{1}{1}{2}{1-({z_p^j} + u+t)}{0,~-{z_p^j} + {y_i^k}-t}{\frac{ {z_p^j} + {x_i^j} \Lambda_i  \mathcal{U}_i   }{{x_i^j}  \mathcal{U}_i \left({y_i^k} + {x_i^j}\mathcal{V}_i \right)^{-1} }}}{\left({y_i^k} + {x_i^j} \mathcal{V}_i  \right)^{-(p-u-\jmath - y_i^k - z_p^j)}}.
\end{align}
Finally, by inserting \eqref{CDFsr_int} and \eqref{V_final} into \eqref{CDFsr}, the coverage probability for SU$_n$ on PC$_k$ can be written as in \eqref{coverage_prob_r}. \qed

\begin{figure}[t]
\centering
\includegraphics[width=0.5 \textwidth]{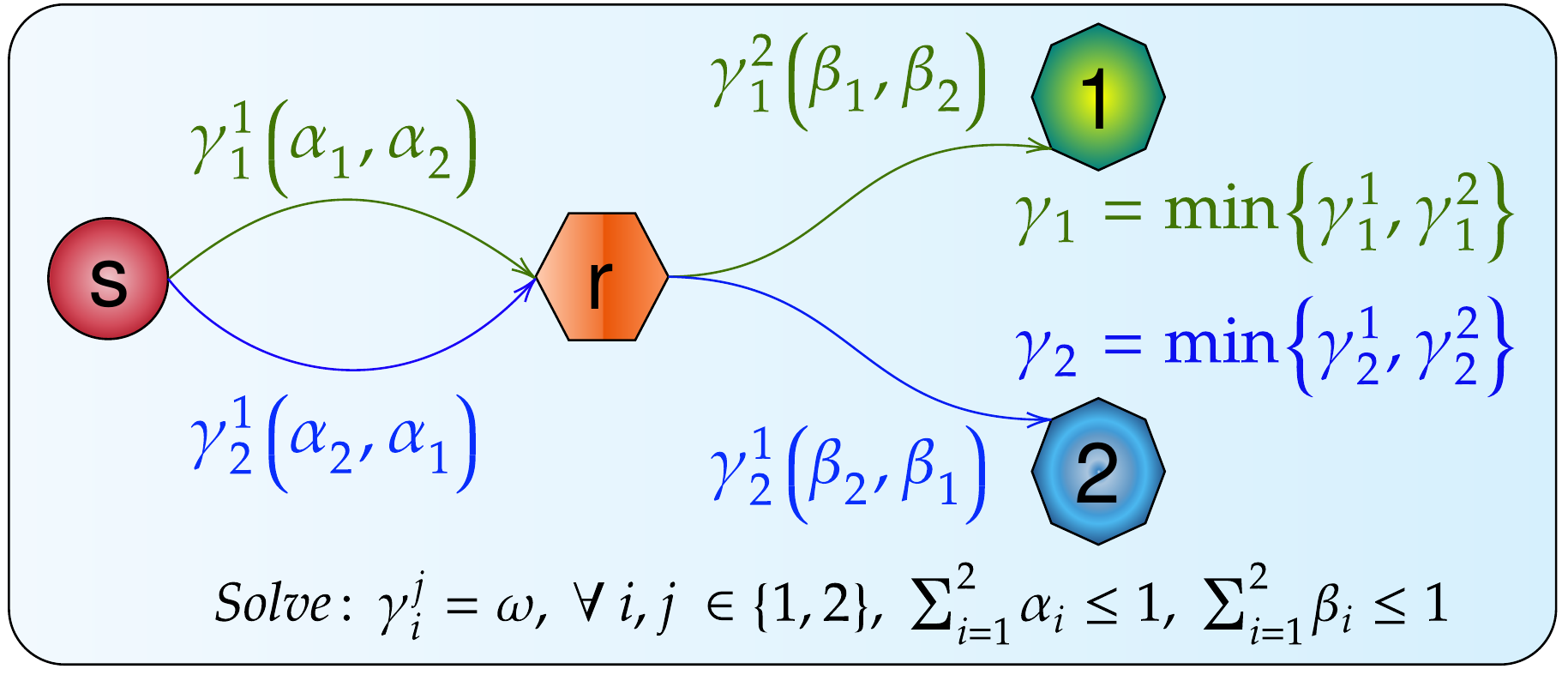}
\caption{Illustration of the two user case.}
\label{fig:2user}
\end{figure} 

\section{Proof of Lemma \ref{lem:SP1}}
\label{app:SP1}
For the sake of clarity of the presentation, let us omit the cluster indices. For simplicity, we consider a two-user case as shown in Fig. \ref{fig:2user}, where $\gamma_i^j$ denotes the SIDNR of SU$_j$ in the phase $i$. For a two-user case, end-to-end SIDNRs, considering the deterministic channel and perfect SIC, are given by $\gamma_1=\min(\gamma_1^1,\gamma_2^1)$ and $\gamma_2=\min(\gamma_1^2,\gamma_2^2)$. Following from Propositions \ref{prop1} and \ref{prop2}, we need to solve the following set of equations to find the optimal power allocation and SIDNRs: 1)$\gamma_i^j= \omega, \forall (i,j)$, 2)$\sum_j \alpha_j \leq \min \left\{ 1,\rm{ITC}_1\right\}$, and 3) $\sum_j \beta_j \leq \min \left\{ 1,\rm{ITC}_2\right\}$, where $\rm{ITC}_1$ and $\rm{ITC}_2$ are the ITC constraints of the first and second phases, respectively. The solution steps are given as follows:

Since the first phase only involves variables $\vect{\alpha}$ and $\omega$, we obtain $\alpha_1=f_1(\omega)$ and $\alpha_2=f_2(\omega)$ from equations $\gamma_1^1 = \omega$ and $\gamma_1^2= \omega$ as follows:
\begin{align}
    \label{eq:alpha1} \alpha_1= \omega \alpha_2\vert \tilde{h}^r_s \vert^2 + \omega I_r = \omega^2 I_r + \omega I_r, \:\: \alpha_2=& \omega I_r.
\end{align}
 Since the second phase only involves variables $\vect{\beta}$ and $\omega$, we obtain $\beta_1=g_1(\omega)$ and $\beta_2=g_2(\omega)$ from equations $\gamma_2^1 = \omega$ and $\gamma_2^2= \omega$ as follows:
\begin{align}
    \label{eq:beta1} \beta_1=& \omega \beta_2  + \omega I_1 = \omega^2 I_2   + \omega I_1,\:\: \beta_2= \omega I_2,
\end{align}
where the definitions of $I_r$ and $I_n$, $\forall n \in\{1,2\} $, is provided in Lemma \ref{lem:SP1}.  
\begin{table}[t]
\centering
\caption{The closed-form optimal power allocations}
\label{tab:closed_form}
\resizebox{0.6\textwidth}{!}{%
\begin{tabular}{|c|c|c|l|}
\hline
\rowcolor[HTML]{9B9B9B} 
{\color[HTML]{FFFFFF} \textbf{C}}                                            & {\color[HTML]{FFFFFF} \textbf{$\accentset{\star}{\vect{\alpha}}$}}                                             & \multicolumn{2}{c|}{\cellcolor[HTML]{9B9B9B}{\color[HTML]{FFFFFF} \textbf{$\accentset{\star}{\vect{\beta}}$}}}                                                                                                       \\ \hline
\cellcolor[HTML]{9B9B9B}{\color[HTML]{FFFFFF} }                                              & {\color[HTML]{333333} $\accentset{\star}{\alpha}_1= I_r \accentset{\star}{\vect{\gamma}} (1+ \accentset{\star}{\vect{\gamma}} )$} & \multicolumn{2}{c|}{{\color[HTML]{333333} $\accentset{\star}{\beta}_1= I_1 \accentset{\star}{\vect{\gamma}} + I_2 \accentset{\star}{\vect{\gamma}} ^2$}}                                                                                \\ \cline{2-4} 
\rowcolor[HTML]{EFEFEF} 
\multirow{-2}{*}{\cellcolor[HTML]{9B9B9B}{\color[HTML]{FFFFFF} \textbf{2}}} & {\color[HTML]{333333} $\accentset{\star}{\alpha}_2= I_r \accentset{\star}{\vect{\gamma}} $}                                       & \multicolumn{2}{c|}{\cellcolor[HTML]{EFEFEF}{\color[HTML]{333333} $\accentset{\star}{\beta}_2=  I_2 \accentset{\star}{\vect{\gamma}}$}}                                                                                                 \\ \hline
\cellcolor[HTML]{9B9B9B}{\color[HTML]{FFFFFF} }                                              & {\color[HTML]{333333} $\accentset{\star}{\alpha}_1= I_r \accentset{\star}{\vect{\gamma}} (1+\accentset{\star}{\vect{\gamma}})^2$} & \multicolumn{2}{c|}{{\color[HTML]{333333} $\accentset{\star}{\beta}_1= I_1 \accentset{\star}{\vect{\gamma}} + I_2 \accentset{\star}{\vect{\gamma}}^2 + I_3 \accentset{\star}{\vect{\gamma}}^2(1 + \accentset{\star}{\vect{\gamma}}) $}} \\ \cline{2-4} 
\cellcolor[HTML]{9B9B9B}{\color[HTML]{FFFFFF} }                                              & {\color[HTML]{333333} $\accentset{\star}{\alpha}_2=  I_r \accentset{\star}{\vect{\gamma}} (1+ \accentset{\star}{\vect{\gamma}})$ }                                                                                                           & \multicolumn{2}{c|}{{\color[HTML]{333333} $\accentset{\star}{\beta}_2= I_2 \accentset{\star}{\vect{\gamma}} + I_3 \accentset{\star}{\vect{\gamma}}^2 $}}                                                                                \\ \cline{2-4} 
\rowcolor[HTML]{EFEFEF} 
\multirow{-3}{*}{\cellcolor[HTML]{9B9B9B}{\color[HTML]{FFFFFF} \textbf{3}}} & {\color[HTML]{333333} $\accentset{\star}{\alpha}_3= I_r \accentset{\star}{\vect{\gamma}} $}                                       & \multicolumn{2}{c|}{\cellcolor[HTML]{EFEFEF}{\color[HTML]{333333} $\accentset{\star}{\beta}_3= I_3 \accentset{\star}{\vect{\gamma}}$}}                                                                                                  \\ \hline
\cellcolor[HTML]{9B9B9B}{\color[HTML]{FFFFFF} }                                              &  {\color[HTML]{333333} $\accentset{\star}{\alpha}_1= I_r \accentset{\star}{\vect{\gamma}} (1+\accentset{\star}{\vect{\gamma}})^3$ }                                                                                                           & \multicolumn{2}{c|}{{\color[HTML]{333333} $\accentset{\star}{\beta}_1= I_1 \accentset{\star}{\vect{\gamma}} + I_2 \accentset{\star}{\vect{\gamma}}^2 + I_3 \accentset{\star}{\vect{\gamma}}^2 (1 + \accentset{\star}{\vect{\gamma}}) + I_4 \accentset{\star}{\vect{\gamma}}^2 (1 +  \accentset{\star}{\vect{\gamma}})^2$ }}                                                                                                                                                                                            \\ \cline{2-4} 
\cellcolor[HTML]{9B9B9B}{\color[HTML]{FFFFFF} }                                              & {\color[HTML]{333333} $\accentset{\star}{\alpha}_2= I_r \accentset{\star}{\vect{\gamma}} (1+\accentset{\star}{\vect{\gamma}})^2$ }                                                                                                           & \multicolumn{2}{c|}{{\color[HTML]{333333} $\accentset{\star}{\beta}_2= I_2 \accentset{\star}{\vect{\gamma}}  + I_3 \accentset{\star}{\vect{\gamma}}^2 + I_4 \accentset{\star}{\vect{\gamma}}^2 (1 + \accentset{\star}{\vect{\gamma}}) $ }}                                                                                                                                                                                            \\ \cline{2-4} 
\cellcolor[HTML]{9B9B9B}{\color[HTML]{FFFFFF} }                                              & {\color[HTML]{333333} $\accentset{\star}{\alpha}_3= I_r \accentset{\star}{\vect{\gamma}} (1+\accentset{\star}{\vect{\gamma}})$ }                                                                                                           & \multicolumn{2}{c|}{{\color[HTML]{333333} $\accentset{\star}{\beta}_3= I_3 \accentset{\star}{\vect{\gamma}} + I_4 \accentset{\star}{\vect{\gamma}}^2 $ }}                                                                                                                                                                                            \\ \cline{2-4} 
\rowcolor[HTML]{EFEFEF} 
\multirow{-4}{*}{\cellcolor[HTML]{9B9B9B}{\color[HTML]{FFFFFF} \textbf{4}}} & {\color[HTML]{333333} $\accentset{\star}{\alpha}_4= I_r \accentset{\star}{\vect{\gamma}}$ }                                                                                                           & \multicolumn{2}{c|}{\cellcolor[HTML]{EFEFEF}{\color[HTML]{333333} $\accentset{\star}{\beta}_4= I_4 \accentset{\star}{\vect{\gamma}} $ }}                                                                                                                                                                    \\ \hline
\end{tabular}%
}
\end{table}

By substituting $\alpha_1=f_1(\omega)$/$\alpha_2=f_2(\omega)$ into $\sum_j \alpha_j \leq \min \left\{ 1,\rm{ITC}_1\right\}$, we can find the optimal SIDNR of the first phase as
\begin{equation}
    \label{eq:omega_1}
    \accentset{\star}{\gamma}_1=  \sqrt{\frac{ \Phi_1}{I_r} +1} -1, 
\end{equation}
where $\Phi_i = \min \left\{ 1,\text{ITC}_i \right\}$, $\forall i \in \{1,2\}$. 
Moreover, substituting $\beta_1=g_1(\omega)$ and $\beta_2=g_2(\omega)$ into $\sum_j \beta_j \leq \min \left\{ 1,\rm{ITC}_2\right\}$, the optimal SIDNR of the second phase can be derived by:
\begin{equation}
    \label{eq:omega_2} 
    \accentset{\star}{\gamma}_2=  \frac{ \sqrt{4 I_2 \Phi_2 + \left(I_1 + I_2 \right)^2 } - I_1 -I_2 }{2 I_2}. 
\end{equation}
Then, the optimal SIDNR can be derived as $\accentset{\star}{\gamma}= \min (\accentset{\star}{\gamma}_1, \accentset{\star}{\gamma}_2)$. Finally, we can obtain $\accentset{\star}{\alpha}_1$-$\accentset{\star}{\alpha}_2$ and $\accentset{\star}{\beta}_1$-$\accentset{\star}{\beta}_2$ by substituting $\accentset{\star}{\gamma}$ into \eqref{eq:alpha1} and \eqref{eq:beta1}, respectively, i.e.,
\begin{align}
    \label{eq:alpha1_star} \accentset{\star}{\alpha}_1= \accentset{\star}{\gamma}^2 I_r + \accentset{\star}{\gamma} I_r, \:\: \accentset{\star}{\alpha}_2=  \accentset{\star}{\gamma} I_r, \:\: \accentset{\star}{\beta}_1= \accentset{\star}{\gamma}^2 I_2   + \accentset{\star}{\gamma} I_1,\:\:  \accentset{\star}{\beta}_2= \accentset{\star}{\gamma} I_2.
\end{align}
By repeating similar steps for $C>2$, we obtain the optimal SIDNRs and power allocations as tabulated in  Table \ref{tab:closed_form}. Based on the observed pattern in Table \ref{tab:closed_form}, the generalizated closed-form equations for the $C$ cluster size can be obtained as in Lemma \ref{lem:SP1}.

\bibliographystyle{ieeetr}
\bibliography{uav_cr_noma}
\end{document}